\documentclass[sigconf]{acmart}

\usepackage{booktabs}
\usepackage{multirow}
\usepackage[table,xcdraw]{xcolor} 
\usepackage{tabularx}
\usepackage{algorithm}
\usepackage{algorithmicx}
\usepackage{algpseudocode}
\usepackage{enumerate}

\usepackage{amsmath}
\usepackage{subfigure}

\newtheorem{remark}{Remark}

\AtBeginDocument{%
  }

\copyrightyear{2026}
\acmYear{2026}
\setcopyright{cc}
\setcctype{by}
\acmConference[CIKM '26]{Proceedings of the 35th ACM International Conference on Information and Knowledge Management}{November 07--11, 2026}{Rome, Italy}
\acmBooktitle{Proceedings of the 35th ACM International Conference on Information and Knowledge Management (CIKM '26), November 07--11, 2026, Rome, Italy}
\acmDOI{10.1145/3799682.3840938}
\acmISBN{979-8-4007-2539-5/2026/11}

\begin{document}

%%
%% The "title" command has an optional parameter,
%% allowing the author to define a "short title" to be used in page headers.

% \section{Title Information} The title of your work should use capital letters appropriately - \url{https://capitalizemytitle.com/} has useful rules for capitalization. Use the {\verb|title|} command to define the title of your work. If your work has a subtitle, define it with the {\verb|subtitle|} command. Do not insert line breaks in your title. If your title is lengthy, you must define a short version to be used in the page headers, to prevent overlapping text. The \verb|title| command has a `short title'' parameter: \begin{verbatim} \title[short title]{full title} \end{verbatim}

\title{RecSoup: Efficient Model Merging Approach \\for Controllable Multi-Objective Recommendation}
\renewcommand{\shorttitle}{RecSoup: Efficient Model Merging Approach for Controllable Multi-Objective Recommendation}
%%
%% The "author" command and its associated commands are used to define
%% the authors and their affiliations.
%% Of note is the shared affiliation of the first two authors, and the
%% "authornote" and "authornotemark" commands
%% used to denote shared contribution to the research.

\author{Zidong Jia}
\affiliation{%
  \department{Gaoling School of Artificial Intelligence}
  \institution{Renmin University of China}
  \city{Beijing}
  \country{China}
}
\email{2025104099@ruc.edu.cn}

\author{Guofu Xie}
\affiliation{%
  \department{Gaoling School of Artificial Intelligence}
  \institution{Renmin University of China}
  \city{Beijing}
  \country{China}
}
\email{2024001045@ruc.edu.cn}

\author{Xiao Zhang}
%\correspondingauthor
\authornote{Xiao Zhang is the corresponding author (e-mail: zhangx89@ruc.edu.cn). }
\affiliation{%
  \department{Gaoling School of Artificial Intelligence}
  \institution{Renmin University of China}
  \city{Beijing}
  \country{China}
}
\email{zhangx89@ruc.edu.cn}

%%
%% By default, the full list of authors will be used in the page
%% headers. Often, this list is too long, and will overlap
%% other information printed in the page headers. This command allows
%% the author to define a more concise list
%% of authors' names for this purpose.
\renewcommand{\shortauthors}{Zidong Jia, Guofu Xie, and Xiao Zhang}
%%
%% The abstract is a short summary of the work to be presented in the
%% article.
\begin{abstract}

Recommender systems often need to balance multiple conflicting objectives. Controllable Multi-Objective Recommendation (CMOR) enables dynamic adjustment of objective trade-offs at inference time without retraining, adapting to shifting user or system requirements. 
Existing CMOR approaches typically rely on output-level reranking or parameter-generating modules, which either limit the scope of feasible trade-offs or incur substantial training and inference cost. 
Model merging provides a promising alternative by reusing multiple objective-specific models via parameter combination. However, its direct application to recommender systems is challenging due to the lack of a shared linear regime and cross-objective interference. 
To address these issues, we propose RecSoup, a lightweight two-stage approach for effective CMOR via specialized model merging. RecSoup first mitigates parameter interference across objectives through a pruning strategy that jointly considers objective-specific importance and cross-objective relevance. It then performs adaptive parameter merging via a learnable preference-to-coefficient mapping, allowing dynamic trade-off control at inference time without joint retraining or additional generative modules. 
Theoretically, we prove that RecSoup outperforms the parameter generation paradigm in effectiveness and efficiency, as it can fully navigate the Pareto front without smoothness constraints. Extensive experiments on multiple benchmarks demonstrate that RecSoup achieves improved Pareto front quality and effective controllability while delivers 165--228× faster inference-time latency than parameter generation approaches. 
Code is available at \url{https://github.com/35798642/RECSOUP.git}.

\end{abstract}

%%
%% The code below is generated by the tool at http://dl.acm.org/ccs.cfm.
%% Please copy and paste the code instead of the example below.
%%
\begin{CCSXML}
<ccs2012>
   <concept>
       <concept_id>10002951.10003317.10003347.10003350</concept_id>
       <concept_desc>Information systems~Recommender systems</concept_desc>
       <concept_significance>500</concept_significance>
       </concept>

 </ccs2012>
\end{CCSXML}

\ccsdesc[500]{Information systems~Recommender systems}

%%
%% Keywords. The author(s) should pick words that accurately describe
%% the work being presented. Separate the keywords with commas.
\keywords{Model Merging, Controllable Recommendation}

% \received{20 February 2007}
% \received[revised]{12 March 2009}
% \received[accepted]{5 June 2009}

%%
%% This command processes the author and affiliation and title
%% information and builds the first part of the formatted document.
\maketitle

\begin{figure}[t]
    \centering
    \includegraphics[width=0.94\linewidth]{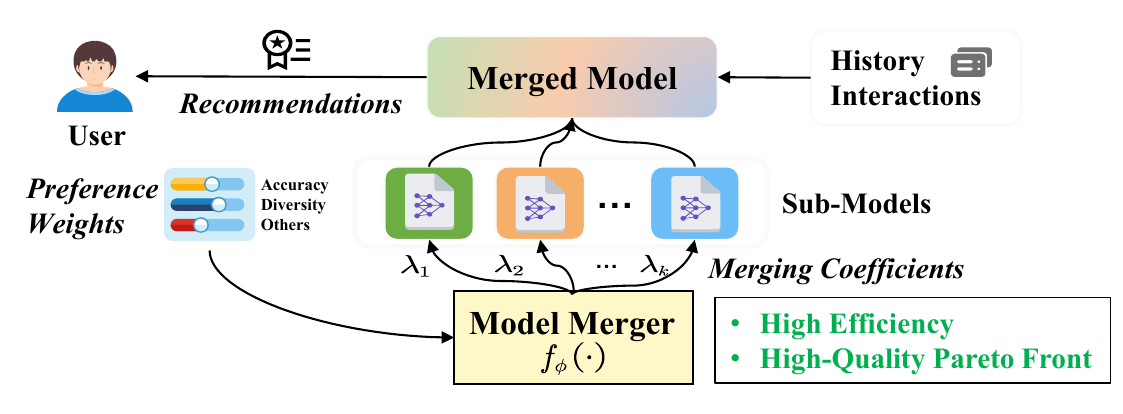}
    \caption{
    \textbf{Illustration of controllable multi-objective recommendation via model merging.}}
    \label{fig:model-merging}
\end{figure}

\section{Introduction}
Modern recommender systems optimize multiple, often conflicting objectives beyond a single metric~\cite{jannach2022multi,zheng2021mors,zaizi2024e-mors,huang2024datamors,ma2023acc-div-moo,ribeiro2012paretohybrid}. While existing multi-objective recommender systems address these trade-offs during training, they typically enforce fixed preferences at inference time~\cite{jannach2023survey,navon2020phn}. In practice, however, the preference weights on objectives are often dynamic, shifting across users' or systems' requirements. This motivates Controllable Multi-Objective Recommendation (CMOR), which enables dynamic adjustment of objective trade-offs at inference time without retraining~\cite{shen2026controllable}.

Research on CMOR can be categorized into two paradigms based on the stage of intervention: post-processing and in-processing approaches. Post-processing approaches, such as rule-based reranking like MMR~\cite{carbonell1998use} intervene by adjusting the initially generated recommendation lists. While being model-agnostic, their performance is constrained by the fixed candidate set, often leading to a narrow scope of trade-off solutions. In-processing focuses on parameter-level modulation, primarily through parameter generation. These parameter generation approaches learn the mapping from user preferences over multiple objectives to the parameter space of recommendation model by employing generative architectures, such as hypernetworks~\cite{navon2020phn,chen2023controllable,hoang2023phnhvi} or diffusion models~\cite{shen2025paragon}. Despite delivering superior performance, these in-processing approaches inflate the training cost with a massive number of additional parameters from the generative modules and incur non-negligible inference latency, while often suffering from optimization instability. 

%引出：已有推荐模型可能不同版本，复用已有的推荐模型，通过模型融合复用；
%Existing CMOR approaches achieve controllability either by modifying outputs or generating model parameters tailored to diverse user preferences, leaving a critical complementary direction largely unexplored: direct manipulation of existing recommendation model parameters.  
%Recent advances in large language models demonstrate that model merging can adjust model behavior through weighted merging of existing model parameters according to multi-objective preferences, without the need for retraining or additional parameter-generating modules~\cite{wortsmanModelSoups2022,yang2024model,ortiz2023task,rameRewardedSoups2023,ilharco2022editing}. 
%Notably, industrial recommendation pipelines often accumulate numerous costly pre-trained, objective-specialized sub-models (e.g., tailored for accuracy or diversity), which are either discarded or require costly joint retraining when adapting to dynamic preferences. 
% 直接引出 
This raises a critical open question: why not reuse existing recommendation models to efficiently achieve controllable trade-offs, rather than constructing new parameters from scratch?
Recent advances in large language models offer a promising answer: \emph{model merging}, which adjusts model behavior through weighted merging of existing model parameters according to multi-objective preferences, eliminating the need for retraining or additional parameter-generating modules~\cite{wortsmanModelSoups2022,yang2024model,ortiz2023task,rameRewardedSoups2023,ilharco2022editing}. 
This paradigm is particularly well-suited for recommender systems, where multiple independently trained, objective-specialized sub-models (e.g., for accuracy or diversity) commonly coexist in industrial training pipelines~\cite{eisenman2022ckpts,cai2022reloop}. Rather than discarding these costly pre-trained checkpoints, model merging offers an efficient mechanism for reusing and combining their expertise. 
As illustrated in Figure~\ref{fig:model-merging}, in this paper, we realize controllable recommendation by dynamically merging objective-specialized sub-models in the parameter space. Given user-specified preference weights across objectives, a lightweight model merger $f_{\phi}(\cdot)$ maps these preferences to merging coefficients that govern the integration of parameters from each sub-model. The resulting merged model directly generates recommendations based on users' historical interactions. Notably, this model merging paradigm achieves both high efficiency and a high-quality Pareto front, while avoiding the costs associated with joint multi-objective retraining. 

However, directly applying existing model merging approaches to recommender systems is non-trivial. Unlike settings where task-specific updates reside in a shared linear regime and can be composed through task arithmetic~\cite{ilharco2022editing,ortiz2023task}, objective-specialized recommendation models exhibit heterogeneous parameter updates driven by structurally conflicting objectives such as accuracy and diversity. As analyzed in Section~\ref{problem_analyses}, recommendation task vectors substantially deviate from the linear regime hypothesis that underpins most task-vector merging approaches. The absence of a shared low-dimensional linear subspace renders naive linear interpolation ineffective and prone to cross-objective interference, motivating the need for a merging strategy tailored to the complex parameter geometry of recommendation models.

% proposed approach
To address these challenges, we propose RecSoup, a two-stage framework for controllable multi-objective recommendation. In the first stage, RecSoup employs a fine-grained pruning strategy to resolve parameter interference. Unlike magnitude-based approaches, we jointly consider objective-specific importance and cross-objective relevance, which establishes a harmonious parameter foundation for merging. In the second stage, RecSoup moves beyond linear interpolation. We integrate a learnable mechanism that captures the mapping from user preference vectors to merging coefficients, allowing the model to dynamically navigate the parameter space. The synergy between pruning-driven conflict mitigation and adaptive merging enables RecSoup to generate high-quality Pareto solutions with superior controllability, bypassing the optimization instability and architectural overhead of generative approaches. 
% contribution
The key contributions are summarized as follows: 
\begin{itemize} 
    \item We present \textbf{RecSoup}, the first framework to integrate model merging into CMOR, adjusting objective trade-offs directly in the parameter space.
    \item We propose a pruning strategy that jointly considers objective-specific and cross-objective importance and can be integrated with existing dynamic merging mechanisms to substantially improve controllable multi-objective recommendation.
    % \item We propose a pruning strategy that jointly considers objective-specific and cross-objective importance, and demonstrate that this scheme can be integrated with existing dynamic merging mechanisms to substantially improve controllable
    % multi-objective recommendation.
    \item We provide theoretical analysis showing that model merging is a more effective and efficient paradigm for CMOR.
    \item Extensive experiments conducted on various backbones and datasets validate that RecSoup achieves higher-quality Pareto fronts while maintaining comparable controllability to state-of-the-art baselines.
\end{itemize}

\section{Related Work}
\emph{Model Merging} integrates task-specific specialists into a unified model without joint retraining, and has been widely studied in LLMs~\cite{yang2024model,matena2022fishermerging,ilharco2022editing,wortsmanModelSoups2022}. Existing methods mainly address interference among task-specific updates, or task vectors, through conflict mitigation~\cite{huang2024emr,lu2024twinmerging,yadav2023ties,yu2024dare,du2024parameter,yang2024surgerymerging,gargiulo2025tsv,marczak2025no,sun2025tatr} or coefficient determination~\cite{yang2023adamerging,du2024parameter}. The former reduces conflicts by pruning redundant or conflicting parameters, such as TIES-merging~\cite{yadav2023ties}, DARE~\cite{yu2024dare}, CABS~\cite{yang2025cabs}, and LwPTV~\cite{zhang2025lwptv}, or by projecting task vectors into orthogonal subspaces, such as TATR~\cite{sun2025tatr}. These methods motivate our pruning design, but their criteria are not tailored to recommendation objectives and overlook objective-specific contribution and cross-objective relevance. 
%Since naive merging often suffers from interference among task-specific updates, hereafter referred to as task vectors, existing methods mainly fall into two categories: conflict mitigation and coefficient determination. Conflict-mitigation methods aim to reduce destructive interference among task vectors. Pruning-based methods, including TIES-merging~\cite{yadav2023ties}, DARE~\cite{yu2024dare}, CABS~\cite{yang2025cabs}, and LwPTV~\cite{zhang2025lwptv}, remove redundant or conflicting parameters based on magnitude, sign consistency, overlap, or layer-wise saliency. TATR~\cite{sun2025tatr} instead mitigates conflicts by projecting task vectors into orthogonal subspaces. These methods motivate our use of pruning, but their criteria are primarily developed in general model-merging settings and do not explicitly model objective-specific contribution or cross-objective relevance in recommendation. 
The latter determines how to combine task vectors. AdaMerging~\cite{yang2023adamerging} adaptively learns merging weights, while PCB-merging~\cite{du2024parameter} considers both intra-task and inter-task balancing, which inspires our pruning design. However, these methods typically yield a single fixed merged model, resulting in a ``one-size-fits-all'' solution~\cite{akiba2025evomerging} that is unsuitable for CMOR, where different preference vectors require different trade-offs among accuracy, diversity, and fairness. Recent Pareto merging methods learn multiple solutions via low-rank structures~\cite{chen2024paretomerging}, surrogate evaluation~\cite{li2025map}, or MoE-based coefficient generation~\cite{tang2024moe}. Following Tang et al.~\cite{tang2024moe}, we adopt a preference-conditioned coefficient generator as determiner and tailor it with our pruning criterion to CMOR.

\emph{Parameter Generation} enables controllability by dynamically generating model weights to meet specific task demands. This approach was pioneered in controllable multi-task learning~\cite{navon2020phn,raychaudhuri2022controllable}, which demonstrated that hypernetworks could map preference vectors to task-specific parameters. Subsequent advances such as PHN~\cite{navon2020phn} and PHN-HVI~\cite{hoang2023phnhvi} have shown that a single hypernetwork can effectively approximate the entire Pareto front. In recommendation, CMR~\cite{chen2023controllable} adapts hypernetwork by generating sensitive re‑ranking parameters for instant preference adjustment, while Hyperbandit and its variant further extend this idea to time‑varying user preferences in streaming recommendation~\cite{shen2023hyperbandit,shen2026enhancing}. Beyond hypernetworks, Paragon~\cite{shen2025paragon} employs diffusion models to synthesize objective-conditioned adapters, matching the performance of retraining at lower cost. Despite their effectiveness in achieving Pareto optimality, parameter generation paradigms face significant challenges regarding architectural complexity and training stability.

\section{Problem Formulation and Analyses}
This section formulates the problem of controllable multi-objective recommendation via model merging and analyzes the limitations of naive task-vector merging in sequential recommendation models.
\begin{figure}[t]
    \centering
\includegraphics[width=0.92\linewidth]{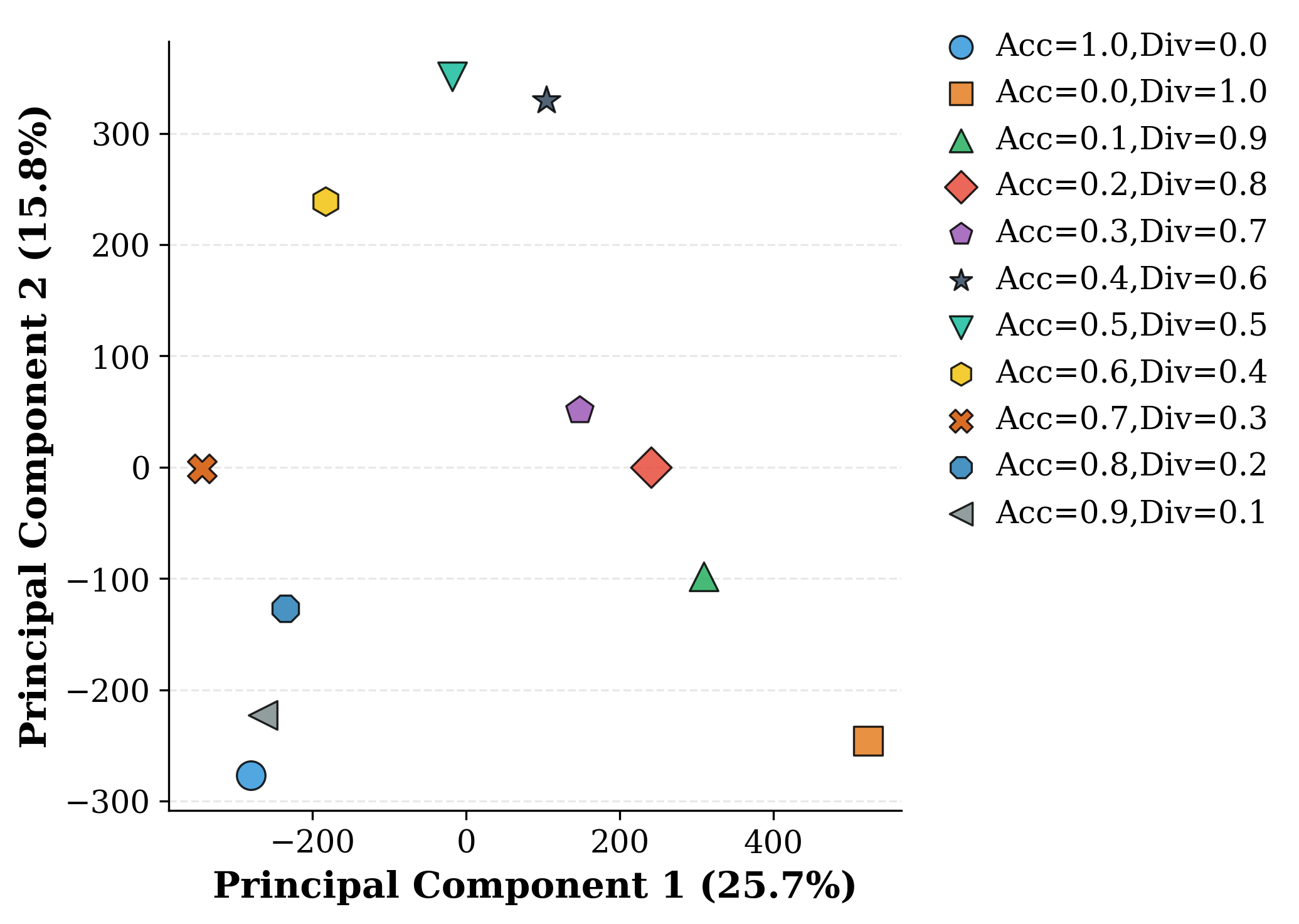} 
   \caption{PCA projection of accuracy- and diversity-optimized task vectors in TiSASRec. The axes correspond to the first (PC1) and second (PC2) principal components, explaining 25.7\% and 15.8\% of the variance, respectively.}
    \label{fig:task_vector_pca}
\end{figure}
\subsection{Problem Formulation}
Let $\mathcal{U}$ and $\mathcal{I}$ denote the sets of users and items, respectively. We consider a multi-objective sequential recommendation scenario with $N$ distinct list-wise evaluation dimensions (e.g., accuracy, diversity). Let $\theta_{\mathrm{base}}$ represent the parameters of a base recommendation model. For each objective $i \in \{1, \dots, N\}$, a specialized expert model $\theta_i$ is obtained, and its corresponding task vector~\cite{ortiz2023task} is defined as 
\begin{equation}
\boldsymbol{\tau}_i = \operatorname{vec}(\theta_i) - \operatorname{vec}(\theta_{\mathrm{base}}),
\label{eq:task_vector}
\end{equation}
where $\operatorname{vec}(\cdot)$ denotes flattening all model parameters into a single row vector. Accordingly, $\boldsymbol{\tau}_i \in \mathbb{R}^{1 \times D}$ captures the objective-specific parameter updates in a linearized parameter space. Throughout this section, all model parameters are discussed in their vectorized form unless otherwise specified.

The goal of controllable multi-objective recommendation is to dynamically adjust model parameters according to a real-time user preference vector $\boldsymbol{\mu} = [\mu_1, \mu_2, \dots, \mu_N]^\top$, which resides in the $(N-1)$-dimensional unit simplex
$\Delta_{N-1} = \{ \boldsymbol{\mu} \in \mathbb{R}_+^N : \sum_{i=1}^N \mu_i = 1 \}$. In the framework of task-vector merging, we aim to learn a mapping function $\mathcal{F}: \boldsymbol{\mu} \mapsto \boldsymbol{\lambda},
$ which transforms the preference vector into a set of merging coefficients $\boldsymbol{\lambda} = [\lambda_1, \dots, \lambda_N]^\top$. Notably, $\boldsymbol{\lambda}$ is not constrained to be identical to $\boldsymbol{\mu}$. The merged model parameters are constructed as
\begin{equation}
    \bar{\theta}(\boldsymbol{\mu}) = \theta_{\mathrm{base}} + \sum_{i=1}^N \lambda_i \boldsymbol{\tau}_i,
    \label{eq:task_merging}
\end{equation}
where $\lambda_i = [\mathcal{F}(\boldsymbol{\mu})]_i$ denotes the merging coefficient associated with the $i$-th task vector.

For a given user $u \in \mathcal{U}$, the recommendation model parameterized by $\bar{\theta}(\boldsymbol{\mu})$ in Eq.~\eqref{eq:task_merging} produces a ranked list of $L$ items, denoted as $\mathcal{L}_u(\bar{\theta}(\boldsymbol{\mu}))$. The performance of this list is evaluated by a set of objective-specific metrics $\{R_i(\mathcal{L}_u)\}_{i=1}^N$. We adopt a linear scalarization to aggregate multiple objectives: 
\begin{equation}
    S(\mathcal{L}_u, \boldsymbol{\mu}) = \sum_{i=1}^N \mu_i \, R_i(\mathcal{L}_u(\bar{\theta}(\boldsymbol{\mu}))),
    \label{eq:weighted_reward}
\end{equation}
which explicitly reflects how the preference vector $\boldsymbol{\mu}$ governs the trade-off among different objectives through the merged parameters. Our problem is to optimize the mapping function $\mathcal{F}$ such that the induced parameter family
$\{ \bar{\theta}(\boldsymbol{\mu}) \mid \boldsymbol{\mu} \in \Delta_{N-1} \}$ approximates the Pareto front of the multi-objective recommendation problem in expectation over users, while ensuring that each metric $R_i$ in Eq.~\eqref{eq:weighted_reward} exhibits consistent and monotonic responsiveness to its corresponding preference weight $\mu_i$.

\subsection{Problem Analyses}
\label{problem_analyses}

\begin{figure}[t]
  \centering
    \includegraphics[width=0.99\linewidth]{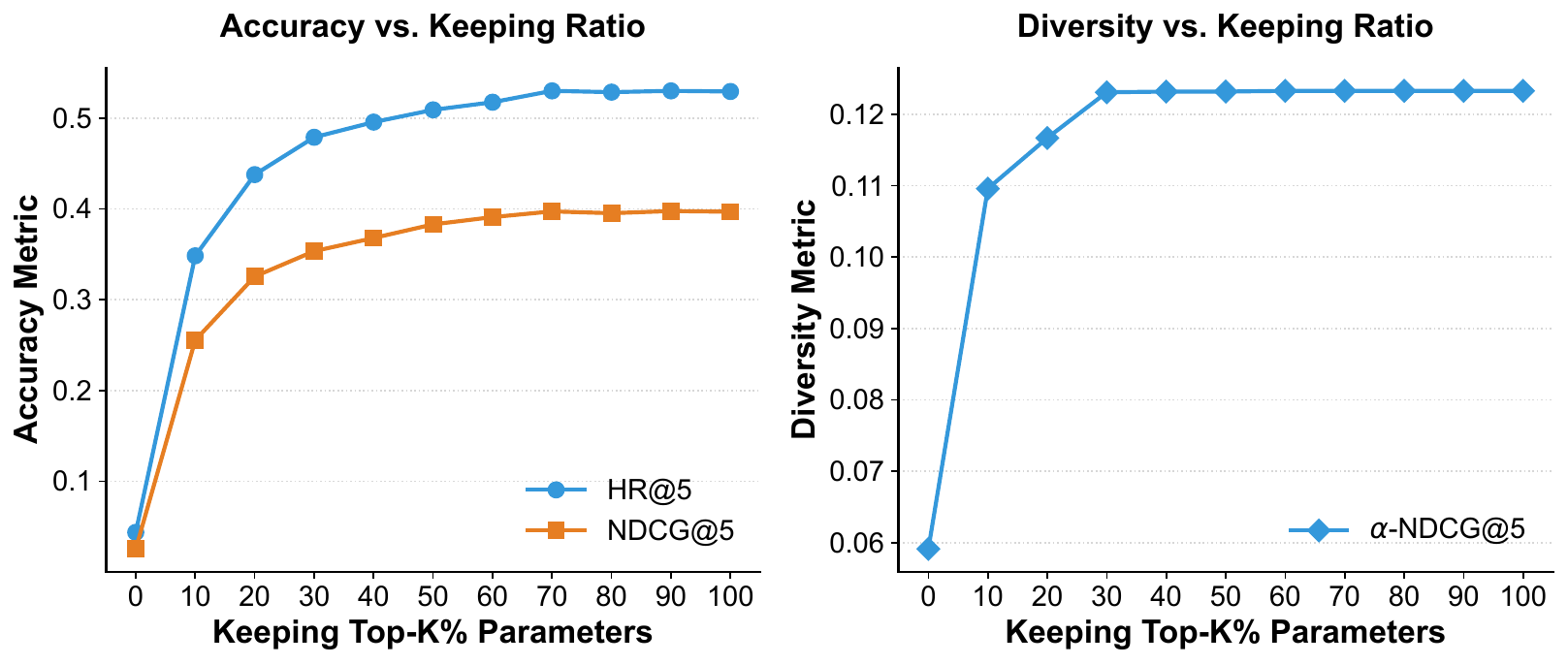}
  \caption{An illustration of the redundant parameters in multi-objective recommendation.}
  \label{fig:redundant_analysis}
\end{figure}

\begin{figure*}[t]
    \centering
    \includegraphics[width=\textwidth]{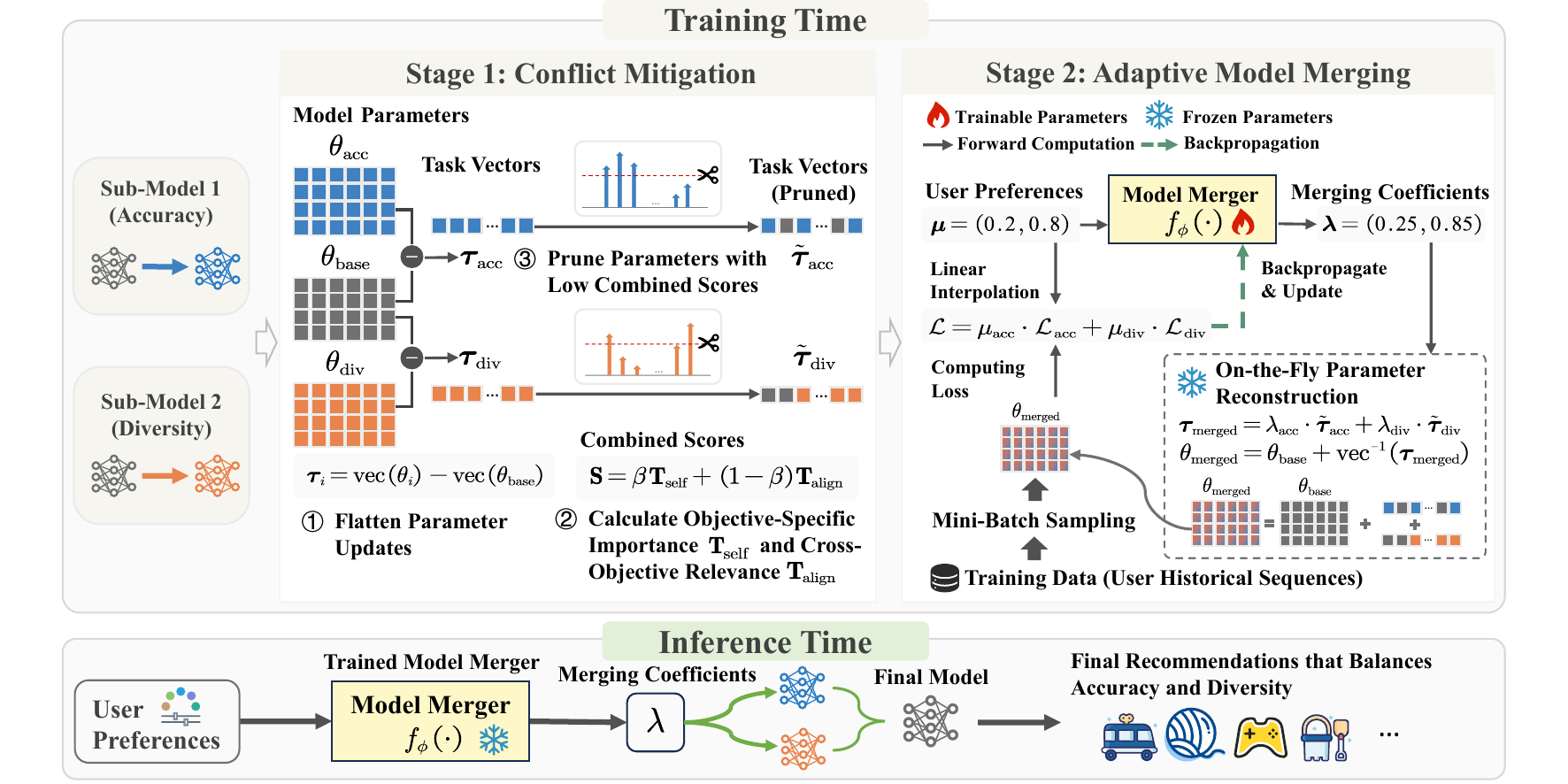}
    \caption{\textbf{The overall architecture of RecSoup.}}
    \label{fig:overall_framework}
\end{figure*}

%In this section, We analyze why task-vector merging approaches, which have shown promising results in large language models, cannot be directly transferred to multi-objective sequential recommendation.

% 为什么naive merging不能用于推荐模型？逻辑链：task arithmetic能够在在llm上work的假设是linear regime，即微调发生在较小的可近似线性的空间 --> 推荐模型，不同目标对应的微调（参数更新）变化较大，不存在这样的linear regime，违反了task arithmetic能够work的假设（PCA实验作为辅助验证，其结果表明不存在共享的线性子空间）--> 所以task arithmetic在推荐模型上不能直接用
% 为什么PCA实验结果能够说明推荐模型违反了linear regime? 背书：不同任务的$(\tau_t)$本质是该子空间内的不同向量组合，反映在PCA上就是：● 所有task vector会聚集在PCA的前1-2个主成分方向（核心线性子空间）；● PC1解释方差极高 即大部分参数变化可被线性子空间解释）
\textbf{Linear Regime Hypothesis.} Task arithmetic (e.g., linear merging of task vectors) has proven effective in models such as CLIP, where its success is grounded in the linear regime hypothesis~\cite{wortsman2022robust,wortsmanModelSoups2022,ilharco2022editing}. This hypothesis assumes that fine-tuning operates within a linear regime, where task-specific parameter updates (task vectors as defined in Eq.~\eqref{eq:task_vector}) are confined to a shared, composable linear subspace consistent with the neural tangent kernel (NTK) approximation. This implies that all task-specific parameter updates are restricted within the linear subspace defined by the NTK tangent space. Different tasks' \(\tau_i\) can therefore be interpreted as distinct vector combinations within this shared subspace.

\textbf{Empirical Violation of the Linear Regime.} To assess whether this assumption extends to sequential recommendation models, we employ Principal Component Analysis (PCA) to characterize the geometric structure of task vectors. As formalized in Ortiz et al.~\cite{ortiz2023task} (Property 2), the linear regime inherently requires that task-specific parameter updates reside within a shared linear subspace (the NTK tangent space). In this scenario, the model's output change after fine-tuning can be approximated by a first-order Taylor expansion (\(f(x;\theta^*) - f(x;\theta_0) \approx \tau^\top \nabla_\theta f(x;\theta_0)\)), which manifests in PCA as a distinctive structural signature: task vectors should exhibit tight clustering along a small number of principal components (PCs), with the majority of parameter variation concentrated in these dominant dimensions. 
%Under the linear regime hypothesis, task-specific parameter updates are expected to concentrate within a shared low-dimensional linear subspace, a property that should manifest as strong variance concentration along a small number of principal components.

However, our PCA results (Figure~\ref{fig:task_vector_pca}) reveals that sequential recommendation models lack this signature of a shared linear subspace. Task vectors optimized for accuracy and diversity are dispersed across disjoint regions of the PCA space, with no evidence of clustering along dominant components: PC1 explains only 25.66\%–26.23\% of the total variance, and the top two components collectively capture less than 45\% of the variance. This fragmented parameter variation directly contradicts the linear regime's requirement: the absence of a dominant principal component indicates no core linear subspace governs cross-objective updates; the unstructured dispersion of task vectors demonstrates that accuracy and diversity rely on fundamentally heterogeneous parameter directions, which cannot be mutually explained through linear combinations; and the low cumulative variance of top components confirms that parameter changes are distributed across high-dimensional space, rather than being confined to a compact linear subspace. 

%Together, these results indicate that the linear regime assumption underlying task arithmetic does not hold in sequential recommendation, thereby limiting the effectiveness of linear task vector merging for controllable multi-objective recommendation.

\textbf{Motivation.} The PCA results (Figure~\ref{fig:task_vector_pca}) show that task vectors for different objectives are dispersed across a high-dimensional space rather than concentrated in a shared linear subspace, indicating a complex task vector manifold that cannot be effectively captured by linear combinations. This motivates the need for \textbf{adaptive merging} to learn preference-conditional coefficients that navigate this manifold. At the same time, sequential recommendation models exhibit heterogeneous parameter importance: although approximately 30\% of parameters can be removed without severe degradation (Figure~\ref{fig:redundant_analysis}), pruning affects objectives differently. Accuracy-oriented updates are highly sensitive, while diversity-oriented ones are robust. This asymmetry suggests that naively merging dense vectors risks suppressing critical parameters and retaining interfering updates, making preference-based merging unstable and less controllable. Together, these observations motivate a two-stage design: first, \textbf{conflict mitigation via pruning} reduces interference and preserves task-critical parameters, and second, \textbf{adaptive merging} explores the complex manifold to achieve controllable multi-objective trade-offs.

\section{RecSoup: The Proposed Approach}
In this section, we present RecSoup, a novel controllable multi-objective recommendation framework that decouples the merging process into two stages: determining which parameters to combine and defining how to weight them. This is realized through two modules--conflict mitigation via pruning and adaptive merging with learnable coefficients.

\subsection{Approach Overview}
% Figure~\ref{fig:overall_framework} illustrates the overall architecture of RecSoup. RecSoup consists of two stages: (1) \textbf{Conflict Mitigation with Pruning.} It transforms dense task vectors into sparse ones by filtering parameters based on their objective-specific importance and cross-objective relevance, thereby reducing the parameter interference inherent in multi-objective recommendation; (2) Adaptive Model Merging with Learnable Coefficients serves to dynamically search for merging coefficients within a learned Pareto manifold. The primary goal of this stage is to find the optimal combination of sub-models that achieves a more balanced and superior trade-off across competing recommendation metrics, such as accuracy and diversity. By first purifying the parameter space and then navigating the manifold for coefficient determination, RecSoup provides a controllable mechanism to balance multiple objectives in recommender systems. At inference time, RecSoup requires one forward pass through the gating network
% to obtain merging coefficients for a given preference vector, enabling efficient and fine-grained control over multi-objective recommendation without retraining. Next, we specify the architecture of RecSoup with details.

Figure~\ref{fig:overall_framework} illustrates RecSoup's two-stage architecture: (1) \textbf{Conflict Mitigation via Pruning.} This stage sparsifies task vectors by filtering parameters based on importance and relevance, thereby mitigating interference. (2) \textbf{Adaptive Model Merging.} This stage dynamically optimizes coefficients within a learned Pareto manifold to balance competing metrics (e.g., accuracy and diversity). By distinguishing parameter selection from coefficient determination, RecSoup ensures robust control. At inference, a gating network generates coefficients for any preference vector in a single forward pass, enabling efficient control without retraining. We detail the architecture below.

\subsection{Conflict Mitigation via Pruning}

The first stage of the RecSoup framework is the \textbf{Sub-model Pruning} step, which aims to generate a set of sparse and specialized task vectors by pruning an initial task vector matrix.
Formally, let $\mathbf{T} \in \mathbb{R}^{N \times D}$ denote the task vector matrix, where the $i$-th row $\mathbf{T}_{i,:}$ corresponds to the task vector associated with the $i$-th objective.
The output of this stage is a pruned task vector matrix $\tilde{\mathbf{T}} \in \mathbb{R}^{N \times D}$.
The pruning procedure is designed to identify and preserve parameters that are highly important to a specific objective while suppressing task-irrelevant updates that cause cross-objective interference. The details are shown in Algorithm~\ref{alg:recsoup_pruning}.

\paragraph{Feature Importance Scoring}
Pruning is conducted independently for each task vector, corresponding to the $i$-th row $\mathbf{T}_{i,:} \in \mathbb{R}^{1 \times D}$ of the task vector matrix.
For each parameter position $j \in \{1, \dots, D\}$, we compute a \textbf{Combined Importance Score} $S_{i,j}$ by integrating two complementary components: objective-specific importance and cross-objective relevance.

\begin{itemize}
    \item \textbf{Objective-Specific Importance.}
    This component measures the relative contribution of parameter $j$ to the overall magnitude of task vector $i$. It is defined as the squared magnitude of the parameter normalized by the $\ell_2$ norm of the task vector:
    \begin{equation}
        T_{\mathrm{self},\,i,j}
        =
        \frac{T_{i,j}^2}{\|\mathbf{T}_{i,:}\|_2^2}.
        \label{eq:self_importance}
    \end{equation}

    \item \textbf{Cross-Objective Relevance.}
    To quantify the degree of coupling between parameter $j$ of task $i$ and other objectives, we compute its interaction with the average update magnitude of the remaining tasks at the same parameter position.
    Specifically, we define $\bar{T}_{i,j} = \frac{1}{N-1} \sum_{k \neq i} T_{k,j}$ and measure the cross-objective relevance as
    \begin{equation}
        T_{\mathrm{align},\,i,j}
        =
        \left|
        T_{i,j} \cdot \bar{T}_{i,j}
        \right|.
        \label{eq:align_score}
    \end{equation}
\end{itemize}

\paragraph{Combined Importance Score}
The final importance score for each parameter position is defined as a weighted combination of the two components:
\begin{equation}
    S_{i,j}
    =
    \beta T_{\mathrm{self},\,i,j}
    +
    (1-\beta)T_{\mathrm{align},\,i,j},
    \label{eq:final_score}
\end{equation}
where $\beta \in [0,1]$ controls the trade-off between objective specialization and cross-objective sensitivity.
We denote $\mathbf{S}_{i,:} = \{ S_{i,j} \}_{j=1}^{D}$ as the combined importance score vector for task $i$.

\paragraph{Top-$k$ Selection via Norm Retention}
After obtaining the combined scores $\mathbf{S}_{i}$, we rank the parameter positions to prioritize those with higher importance.
Let $\pi$ be the permutation of indices $\{1, \dots, D\}$ such that the scores are sorted in descending order: $S_{i, \pi(1)} \ge S_{i, \pi(2)} \ge \dots \ge S_{i, \pi(D)}$.
Rather than enforcing a fixed parameter count, we adopt a dynamic norm retention strategy.
We determine the minimal number of top-ranked parameters, denoted as $k_i$, required to preserve a predefined fraction $\rho$ of the task's total information (measured by self-importance):
\begin{equation*}
    k_i = \min \left\{ k \;\bigg|\; \frac{\sum_{m=1}^{k} T_{\mathrm{self},\,i,\,\pi(m)}}{\sum_{j=1}^{D} T_{\mathrm{self},\,i,\,j}} \ge \rho \right\}.
    \label{eq:cutoff_k}
\end{equation*}
Accordingly, the selected index set is defined as the top-$k_i$ indices: $\mathcal{K}_i = \{ \pi(1), \dots, \pi(k_i) \}$.
The pruned task vector is then generated via element-wise masking: $ \tilde{\mathbf{T}}_{i}=\mathbf{T}_{i} \odot\mathbb{I}_{\mathcal{K}_i}$, where $\mathbb{I}_{\mathcal{K}_i} \in \{0,1\}^{D}$ is the indicator vector for the set $\mathcal{K}_i$.
This strategy ensures that the pruned vector retains the core semantic information (high norm retention) while effectively filtering out conflicting or negligible parameters (low $S$ score), thereby producing a refined subspace for stable merging.

\textbf{Rationale for Absolute Value in Cross-Objective Relevance.}
% We intentionally employ the absolute value in Eq.~\eqref{eq:align_score} to relax the directional consistency assumptions commonly adopted in task arithmetic approaches.
% In multi-objective recommendation, parameters that are highly relevant to multiple objectives may naturally exhibit opposite update directions—for example, improving accuracy may inherently reduce diversity. 
% By taking the absolute interaction magnitude, we measure cross-objective relevance in terms of coupling strength rather than directional alignment.
% This design enables RecSoup to preserve parameters that encode critical yet competing objective-specific information, instead of erroneously pruning them due to sign conflicts.
We intentionally employ the absolute value in Eq.~\eqref{eq:align_score} to relax the directional consistency assumptions typically imposed by task arithmetic approaches. In multi-objective recommendation, parameters critical to multiple tasks often exhibit conflicting update directions. For instance, improving accuracy may inherently reduce diversity. By focusing on the magnitude of the interaction, we measure cross-objective relevance based on coupling strength rather than directional alignment. This design ensures that RecSoup preserves parameters encoding essential yet competing information, preventing their erroneous pruning due solely to sign conflicts.

% \paragraph{Top-$k$ Selection via Norm Retention}
% After obtaining the combined scores $\mathbf{S}_{i,:}$, we sort parameter positions in descending order of $S_{i,j}$.
% Rather than enforcing a fixed sparsity ratio, we determine the cutoff index $k_i$ by ensuring that the retained parameters preserve a predefined fraction $\rho$ of the original task vector norm:

% \begin{equation*}
%     \mathcal{K}_i
%     =
%     \min \left\{
%     \mathcal{K}
%     \;\bigg|\;
%     \sum_{j \in \mathcal{K}}
%     T_{\mathrm{self},\,i,j}
%     \ge \rho
%     \right\}.
% \end{equation*}
% The pruned task vector is then obtained by element-wise masking:
% \begin{equation*}
%     \tilde{\mathbf{T}}_{i,:}
%     =
%     \mathbf{T}_{i,:}
%     \odot
%     \mathbb{I}_{\mathcal{K}_i},
% \end{equation*}
% where $\mathbb{I}_{\mathcal{K}_i} \in \{0,1\}^{1 \times D}$ is the indicator vector corresponding to the selected index set $\mathcal{K}_i$.
% This norm-retention strategy ensures that each pruned task vector preserves its core semantic information while eliminating noisy or conflicting parameter updates, thereby producing a purified subspace that facilitates stable preference-based merging.

\begin{algorithm}[t]
\caption{Pruning Procedure}
\label{alg:recsoup_pruning}
\begin{algorithmic}[1]
\renewcommand{\algorithmicrequire}{\textbf{Input:}}
\renewcommand{\algorithmicensure}{\textbf{Output:}}
\Require Task vectors $\mathbf{T} \in \mathbb{R}^{N \times D}$, norm retention ratio $\rho$, weight $\beta$
\Ensure Pruned task vectors $\tilde{\mathbf{T}} \in \mathbb{R}^{N \times D}$
\For{each task $i = 1, \dots, N$}
    \Statex $\triangleright$ Compute Objective-Specific Importance: 
    \State $T_{\mathrm{self},\,i,:} \leftarrow \mathbf{T}_{i,:} \odot \mathbf{T}_{i,:} / \|\mathbf{T}_{i,:}\|_2^2$
    \Statex $\triangleright$ Compute Cross-Objective Relevance:
    \State $\bar{\mathbf{T}}_{i,:} \leftarrow \frac{1}{N-1} \sum_{k \neq i} \mathbf{T}_{k,:}$
    \State $T_{\mathrm{align},\,i,:} \leftarrow |\mathbf{T}_{i,:} \odot \bar{\mathbf{T}}_{i,:}|$
    \Statex $\triangleright$ Compute Combined Score: 
    \State $\mathbf{S}_{i,:} \leftarrow \beta T_{\mathrm{self},\,i,:} + (1-\beta) T_{\mathrm{align},\,i,:}$
    \State Select top-$k_i$ indices $\mathcal{K}_i$ s.t. $\sum_{j \in \mathcal{K}_i} T_{\mathrm{self},\,i,j} \ge \rho$
    \State Prune task vector: $\tilde{\mathbf{T}}_{i,:} \leftarrow \mathbf{T}_{i,:} \odot \mathbb{I}_{\mathcal{K}_i}$
\EndFor
\State \Return $\tilde{\mathbf{T}}$
\end{algorithmic}
\end{algorithm}

% \begin{algorithm}[t]
% \caption{Pruning Procedure}
% \label{alg:recsoup_pruning}
% \begin{algorithmic}[1]
% \renewcommand{\algorithmicrequire}{\textbf{Input:}}
% \renewcommand{\algorithmicensure}{\textbf{Output:}}
% \Require Task vectors $\mathbf{T} = \{\mathbf{T}_i\}_{i=1}^{N}$, norm retention ratio $\rho$, weights $\beta_1, \beta_2$
% \Ensure Pruned task vectors $\tilde{\mathbf{T}} = \{\tilde{\mathbf{T}}_i\}_{i=1}^{N}$
% \Statex
% \For{each task $i = 1, \dots, N$}
%     \Statex $\triangleright$ Compute Objective-Specific Importance:
%     \State $\mathbf{T}_{\text{self}}^{(i)} \leftarrow \left\{ \frac{\mathbf{T}_{ij}^2}{\|\mathbf{T}_i\|_2^2} \right\}_{j=1}^{D}$
    
%     \Statex $\triangleright$ Compute Cross-Objective Relevance:
%     \State $\bar{\mathbf{T}} \leftarrow \frac{1}{N-1} \sum_{k \neq i} \mathbf{T}_k$
%     \State $\mathbf{T}_{\text{align}}^{(i)} \leftarrow \left\{ |\mathbf{T}_{ij} \cdot \bar{\mathbf{T}}_j| \right\}_{j=1}^{D}$
    
%     \Statex $\triangleright$ Compute Combined Importance Score:
%     \State $\mathbf{S}_i \leftarrow \beta \mathbf{T}_{\text{self}}^{(i)} + (1-\beta) \mathbf{T}_{\text{align}}^{(i)}$
    
%     \Statex $\triangleright$ Select top-$k_i$ dimensions by $\mathbf{S}_i$ such that:
%     \State $\mathcal{K}_i \leftarrow \min \{ \mathcal{K} \mid \sum_{j \in \mathcal{K}} \mathbf{T}_{\text{self},j}^{(i)} \geq \rho \}$
    
%     \Statex $\triangleright$ Prune task vector:
%     \State $\tilde{\mathbf{T}}_i \leftarrow \mathbf{T}_i \odot \mathbb{I}_{\mathcal{K}_i}$
% \EndFor
% \State \Return $\tilde{\mathbf{T}}$
% \end{algorithmic}
% \end{algorithm}

\subsection{Adaptive Model Merging with Learnable Coefficients}
While our learnable coefficients modules adopts a preference-conditional architecture similar to PWE MoE~\cite{tang2024moe}, RecSoup fundamentally redefines the \textit{basis construction} for the gating network. Existing approaches operate directly on dense, high-interference task vectors, forcing the model to simultaneously handle parameter conflicts and objective trade-offs. In contrast, RecSoup decouples these challenges: it performs conflict-aware pruning \textit{prior} to coefficient learning. This ensures that the gating network navigates a purified parameter subspace, allowing it to focus on learning precise coordinates on the Pareto manifold rather than compensating for noise or interference inherent in dense vectors.

After the pruning stage, RecSoup produces the pruned task vector matrix
$\tilde{\mathbf{T}} \in \mathbb{R}^{N \times D}$, where each row
$\tilde{\mathbf{T}}_{i,:} \in \mathbb{R}^{1 \times D}$ represents the task vector of an
objective-specialized sub-model with reduced parameter interference.
These task vectors form the basis for adaptive model merging in the second stage. Following~\cite{tang2024moe}, we adopt a lightweight
gating network $f_\phi(\cdot)$,  which we refer to as the model merger, that learns a mapping from user preference vectors
$\boldsymbol{\mu} \in \mathbb{R}^N$ to approximately optimal merging coefficients,  \begin{equation}
    \boldsymbol{\lambda} = f_\phi(\boldsymbol{\mu}),
\label{eq:mapping}
\end{equation}
where each dimension of $\mathbf{\mu}$ encodes the relative importance of an optimization objective,
such as accuracy or diversity. 
The output coefficients $\boldsymbol{\lambda}\in \mathbb{R}^N$ determine how the pruned task vectors are combined. The merged task vector is obtained via weighted aggregation: $ \boldsymbol{\tau}_{\mathrm{merged}} = \sum_{i=1}^N \lambda_i \tilde{\mathbf{T}}_{i,:}$, which is added to the base model parameters to recover the full merged model weights,  $\theta_{\mathrm{merged}} = \theta_{\mathrm{base}} + \boldsymbol{\tau}_{\mathrm{merged}}$. During training, only the gating network parameters $\phi$ are updated, while the pruned task vectors
$\tilde{\mathbf{T}}$ and the base model weights $\theta_{\mathrm{base}}$ remain fixed.
To reduce computation, we evaluate the preference-weighted multi-objective loss using a small subset of the training data:
\begin{equation*}
    \mathcal{L}(\theta_{\mathrm{merged}}, \boldsymbol{\mu}) = \sum_{i=1}^{N} \mu_i \cdot \mathcal{L}_i(\theta_{\mathrm{merged}}),
\end{equation*}
where $\mathcal{L}_i$ is the task-specific loss for the $i$-th objective.

% RecSoup's merging stage is structurally related to PWE MoE~\cite{tang2024moe}, as both map user preferences to task vector coefficients. However, RecSoup differs fundamentally in its basis construction: while PWE MoE operates on dense, high-interference task vectors, RecSoup performs conflict-aware pruning prior to coefficient learning. By decoupling parameter selection from weight allocation, the gating network navigates a purified subspace with reduced noise. This allows the model to learn precise coordinates on a pruned Pareto manifold rather than compensating for parameter conflicts during merging.

% The success of RecSoup lies in its two-stage separation: pruning identifies the important parameter support to minimize inter-task interference, while the gating network calibrates the relative influence based on preferences. This decoupling significantly reduces the optimization burden, enabling a lightweight architecture to achieve superior trade-offs. The synergy between these stages—where pruning conditions the parameter space for more stable navigation—is empirically validated in Section~\ref{sec:synergy}.

The efficacy of RecSoup stems from its two-stage decoupling strategy: pruning first identifies the critical parameter support to minimize inter-task interference, after which the gating network calibrates the relative influence of each task based on user preferences. This separation significantly alleviates the optimization burden, allowing a lightweight architecture to achieve superior trade-offs. The synergy between the two stages is empirically validated in Section~\ref{sec:synergy}, where pruning refines the parameter space and facilitates more stable and effective exploration during adaptive merging.

\section{Discussion}
%{\bf Comparison with parameter generation approaches.}
\subsection{Theoretical Analysis}
This section theoretically compares model merging and parameter generation to show the former's advantage in controllable multi-objective recommendation. 
Let \( \Theta \subseteq \mathbb{R}^n \) be the convex parameter space of the recommendation model (where \( n \) denotes the number of model parameters). Without loss of generality, we consider a two-objective scenario: let \( \theta_1, \theta_2 \in \Theta \) be two single-objective optimal expert models optimized for two conflicting objectives (e.g., accuracy and diversity); let \( P \subseteq \Theta \) be the {\bf Pareto-optimal set} for the controllable multi-objective recommender system, which is defined as the convex hull of \( \theta_1 \) and \( \theta_2 \) (i.e., \( P = \text{conv}(\{\theta_1, \theta_2\}) \), consisting of all convex combinations \( \theta = \lambda \theta_1 + (1-\lambda)\theta_2 \) with \( \lambda \in [0,1] \)), which is a common assumption in multi-objective model merging \cite{rameRewardedSoups2023}; let \( L_0 = \|\theta_1 - \theta_2\|_2 > 0 \) be the {\bf diameter} of \( P \) (Euclidean distance between the two expert models). 
Let \( g: [0,1] \to \Theta \) denote the {\bf parameter generation mapping} (e.g., hypernetwork \cite{chen2023controllable}) in parameter generation approaches, which is designed to map the user preference \( \lambda \in [0,1] \) to model parameters for multi-objective recommendation. Let \( L > 0 \) be the {\bf Lipschitz constant} of \( g \), satisfying \( \|g(\lambda_1) - g(\lambda_2)\|_2 \leq L \cdot |\lambda_1 - \lambda_2| \) for all \( \lambda_1, \lambda_2 \in [0,1] \).
The following theorem clarifies the relationship between the optimal solution space of parameter generation approaches and that of the model merging paradigm in RecSoup. 
Proof is provided in Appendix~\ref{app:proof}.
\begin{theorem}
\label{thm:recmerge:solutionspace}
   If \( L < L_0 \), then the parameter generation mapping \( g \) in parameter generation approaches cannot generate all elements of the Pareto-optimal set \( P = \mathrm{conv}(\{\theta_1, \theta_2\}) \) achievable via model merging in RecSoup.
\end{theorem}
\begin{remark}
    Theorem~\ref{thm:recmerge:solutionspace} illustrates that parameter generation approaches are constrained by the Lipschitz smoothness of g: when \( L < L_0 \), the mapping g is too ``smooth'' to cover the full Pareto-optimal set induced by model merging. It highlights model merging's advantage in controllable multi-objective recommendation: unlike parameter generation, merging avoids such smoothness limitations and can fully navigate the Pareto front, enabling more flexible trade-off control.
\end{remark}

% inference time 时间复杂度，点出我们的方法recsoup
\subsection{Time Complexity Analysis}
\label{sec:complexity analyses}
We analyze the inference-time complexity of parameter generation approaches and RecSoup, considering both adapter-only and full-parameter controllable settings. Using SASRec as a representative backbone, let $d$ denote the transformer hidden dimension, $L$ the number of layers, $H$ the hidden width of the hypernetwork (typically $O(d)$), $|\mathcal{I}|$ the number of items, $T$ the number of denoising steps in diffusion-based generation, $N$ the number of objectives, and $P = L d^2 + |\mathcal{I}| d$ the full-parameter size. 

For parameter generation approaches, task-specific parameters are dynamically generated at inference time. Under adapter-only setting, each adapter consists of two linear projections with $O(d^2)$ parameters. Hypernetwork-based approaches~\cite{chen2023controllable} employ an MLP to generate adapter parameters, resulting in an inference-time complexity $T_{\mathrm{hypernet}} = O(d^2 \cdot H)$, while diffusion-based approaches~\cite{shen2025paragon} iteratively refine parameters over $T$ steps, yielding %$T_{\mathrm{diffusion}} = %
$O(T \cdot d^3)$. When controllability is extended to all model parameters, the generation cost scales with the full parameter size. Accordingly, hypernetwork- and diffusion-based approaches incur %$T_{\mathrm{hypernet}}^{\mathrm{full}} =$
$O\big(P\cdot H\big)$ and 
%T_{\mathrm{diffusion}}^{\mathrm{full}} = 
$O\big(T\cdot P\cdot d\big)$, respectively. 

RecSoup follows a model merging paradigm, where task-specific parameters are pre-trained and merged at inference time. In the adapter-only setting, merging $N$ adapters requires %$T_{\mathrm{RecSoup}} = 
$O\big(N \cdot d^2\big)$, while a lightweight gating network maps an $N$-dimensional preference vector to merging coefficients (Eq.~\eqref{eq:mapping}) with negligible cost $O(N^2)$. In the full-parameter setting, RecSoup directly merges full models, yielding %$T_{\mathrm{RecSoup}}^{\mathrm{full}} = 
$O\big(N \cdot P\big)$. 

Compared to parameter-generation approaches, RecSoup eliminates auxiliary networks and iterative refinement at inference time. Consequently, it avoids the cubic dependence on $d$ introduced by hypernetwork- and diffusion-based approaches, achieving quadric or linear scaling with respect to the parameter size and the number of objectives in different controllable settings. Although full-parameter merging can be dominated by the item embedding size when $|\mathcal{I}| \gg d^2$, RecSoup remains substantially more efficient than diffusion-based generation and offers a favorable trade-off between inference efficiency and controllability.

% 去掉，可以把实验中的分散到related work或者是结论
% \subsection{Difference with Existing Merging Approaches}
% Compared with state-of-the-art merging approaches, RecSoup distinguishes itself in the following aspects:
% (i) \textbf{TIES-merging vs. RecSoup}: While TIES-merging focuses on magnitude-based pruning and sign alignment, RecSoup formulates pruning as a multi-objective optimization. By incorporating absolute cross-objective relevance scores as defined in Eq.~(\ref{eq:align_score}), RecSoup preserves critical features in the shared parameter space that simple magnitude selection might overlook.
% (ii) \textbf{TATR vs. RecSoup}: TATR restricts updates to an orthogonal ``trust region'' to minimize interference. Conversely, RecSoup recognizes that parameters in all directions can conflict in recommendation contexts; thus, it explicitly balances intra-task importance and cross-task utility to selectively retain features rather than merely restricting the update space.
% (iii) \textbf{PCB-Merging vs. RecSoup}: Unlike PCB-Merging which employs fixed factors for merging, RecSoup applies task-relationship logic to generate optimal pruning masks. More importantly, RecSoup utilizes a learnable preference-to-weight mapping defined in Eq.~(\ref{eq:mapping}) instead of fixed coefficients, enabling controllable navigation
% of the Pareto front to adapt to shifting user preferences without retraining.

\section{Experiments}
We conducted experiments across different backbones to evaluate the performance of RecSoup.
\subsection{Experimental Settings}
\subsubsection{Dataset.}We evaluate \text{RecSoup} on three benchmarks: \textbf{MovieLens-1M}\footnote{\url{https://grouplens.org/datasets/movielens/1m/}}, and two categories from the \textbf{Amazon Review} dataset\footnote{\url{https://jmcauley.ucsd.edu/data/amazon/}} (Grocery and Gourmet Food, Toys and Games). We follow the pre-processing procedures adopted in prior work~\cite{shen2025paragon} to ensure consistency in data filtering and sequence organization.

\begin{table*}
  \caption{Overall performance of all approaches. HV = Hypervolume, LF = Length of Pareto Front, Contro. = Controllability. ``w/o Prune'' omits the pruning stage in RecSoup, while ``w/o MLP'' uses an identity mapping in RecSoup. \textbf{Bold} and \underline{underlined} indicate the best and second-best among main competing approaches. $\dagger$ denotes oracle results (Retrain), shown for reference.}
  \label{tab:overall_performance}
  \centering
  \renewcommand{\arraystretch}{0.9}
  \begin{tabular}{>{\arraybackslash}p{1.4cm} >{\arraybackslash}p{1.4cm}
  >{\centering\arraybackslash}p{1.3cm} c >{\centering\arraybackslash}p{1.3cm}
  >{\centering\arraybackslash}p{1.3cm} c >{\centering\arraybackslash}p{1.3cm}
  >{\centering\arraybackslash}p{1.3cm} c >{\centering\arraybackslash}p{1.3cm} }
    \toprule
     \multirow{2}{*}{\textbf{Backbone}} & \multirow{2}{*}{\textbf{Approach}} &
     \multicolumn{3}{c}{MovieLens-1M} &
     \multicolumn{3}{c}{Grocery and Gourmet Food} &
     \multicolumn{3}{c}{Toys and Games} \\
     \cmidrule(lr){3-5} \cmidrule(lr){6-8} \cmidrule(lr){9-11}
     & & HV $\uparrow$ & LF $\uparrow$ & Contro. $\uparrow$
       & HV $\uparrow$ & LF $\uparrow$ & Contro. $\uparrow$
       & HV $\uparrow$ & LF $\uparrow$ & Contro. $\uparrow$ \\
     \midrule
% ---------------- SASRec ----------------
\multirow{9}{*}{\centering SASRec}
& MMR      & 0.1153 & 2  & 0.3333 & 0.5323 & 5  & 1.0000 & \textbf{0.5762} & 7  & 0.8889 \\
& Paragon  & 0.5362 & 11 & 1.0000 & \underline{0.6013} & 10 & 0.9818 & 0.2479 & 8  & 0.9091 \\
& CMR      & 0.5290 & 10 & 0.9818 & 0.1269 & 8  & 0.8545 & 0.3820 & 11 & 1.0000 \\
& Soup     & \underline{0.6536} & 7  & 0.8182 & 0.5404 & 10 & 0.9818 & 0.5467 & 11 & 1.0000 \\
& TATR     & 0.6227 & 6  & 0.7636 & 0.4831 & 9  & 0.9636 & 0.4997 & 11 & 1.0000 \\
& \cellcolor{gray!10}\textbf{RecSoup}
           & \cellcolor{gray!10}\textbf{0.7090} & \cellcolor{gray!10}8  & \cellcolor{gray!10}0.9091
           & \cellcolor{gray!10}\textbf{0.6088} & \cellcolor{gray!10}9  & \cellcolor{gray!10}0.8910
           & \cellcolor{gray!10}\underline{0.5755} & \cellcolor{gray!10}11 & \cellcolor{gray!10}1.0000 \\
& w/o MLP  & 0.6675 & 8  & 0.9091 & 0.5995 & 10 & 0.9818 & 0.5416 & 11 & 1.0000 \\
& w/o Prune& 0.7228 & 11 & 1.0000 & 0.5437 & 11 & 1.0000 & 0.5594 & 11 & 1.0000 \\
\cmidrule(lr){2-11}
& Retrain$^\dagger$
           & 0.6604 & 9  & 1.0000 & 0.7358 & 11 & 1.0000 & 0.6120 & 10 & 0.8485 \\
\midrule
% ---------------- GRU4Rec ----------------
\multirow{9}{*}{\centering GRU4Rec}
& MMR      & 0.1076 & 2  & 0.5333 & \textbf{0.4653} & 3  & 0.6889 & \textbf{0.5428} & 9  & 1.0000 \\
& Paragon  & 0.4143 & 8  & 0.8910 & \underline{0.3699} & 6 & 0.8910 & 0.0521 & 6  & 0.6909 \\
& CMR      & 0.4464 & 11 & 1.0000 & 0.0027 & 1  & 0.0364 & 0.3577 & 11 & 1.0000 \\
& Soup     & \underline{0.5869} & 8 & 0.8910 & 0.2566 & 11 & 1.0000 & 0.4211 & 11 & 1.0000 \\
& TATR     & 0.5716 & 8  & 0.8910 & 0.2995 & 11 & 1.0000 & 0.4131 & 10 & 1.0000 \\
& \cellcolor{gray!10}\textbf{RecSoup}
           & \cellcolor{gray!10}\textbf{0.6497} & \cellcolor{gray!10}11 & \cellcolor{gray!10}1.0000
           & \cellcolor{gray!10}0.3305 & \cellcolor{gray!10}11 & \cellcolor{gray!10}1.0000
           & \cellcolor{gray!10}\underline{0.4315} & \cellcolor{gray!10}11 & \cellcolor{gray!10}1.0000 \\
& w/o MLP  & 0.6051 & 8  & 0.8940 & 0.2734 & 10 & 1.0000 & 0.4211 & 11 & 1.0000 \\
& w/o Prune& 0.5892 & 7  & 0.8727 & 0.2710 & 11 & 1.0000 & 0.4315 & 11 & 1.0000 \\
\cmidrule(lr){2-11}
& Retrain$^\dagger$
           & 0.6569 & 10 & 0.9818 & 0.6999 & 10 & 1.0000 & 0.6073 & 8  & 0.9091 \\
\midrule
% ---------------- TiSASRec ----------------
\multirow{9}{*}{\centering TiSASRec}
& MMR      & 0.1194 & 2  & 0.4222 & \textbf{0.5115} & 5  & 1.0000 & 0.5967 & 8  & 1.0000 \\
& Paragon  & 0.5228 & 10 & 0.9818 & 0.1220 & 11 & 1.0000 & 0.2347 & 8  & 0.9455 \\
& CMR      & 0.5299 & 11 & 1.0000 & 0.1000 & 7  & 0.8909 & 0.4319 & 11 & 1.0000 \\
& Soup     & \underline{0.6110} & 6 & 0.9636 & 0.4628 & 11 & 1.0000 & \underline{0.6110} & 10 & 0.9818 \\
& TATR     & 0.5558 & 8  & 0.8728 & 0.3893 & 11 & 1.0000 & 0.5537 & 10 & 1.0000 \\
& \cellcolor{gray!10}\textbf{RecSoup}
           & \cellcolor{gray!10}\textbf{0.6902} & \cellcolor{gray!10}11 & \cellcolor{gray!10}1.0000
           & \cellcolor{gray!10}\underline{0.4843} & \cellcolor{gray!10}9 & \cellcolor{gray!10}0.9455
           & \cellcolor{gray!10}\textbf{0.6148} & \cellcolor{gray!10}9 & \cellcolor{gray!10}0.9636 \\
& w/o MLP  & 0.6128 & 7  & 1.0000 & 0.4695 & 11 & 1.0000 & 0.6120 & 9  & 0.9636 \\
& w/o Prune& 0.6722 & 11 & 1.0000 & 0.4735 & 9  & 0.9272 & 0.6126 & 9  & 1.0000 \\
\cmidrule(lr){2-11}
& Retrain$^\dagger$
           & 0.6755 & 9  & 0.9818 & 0.6387 & 10 & 0.9818 & 0.6078 & 11 & 1.0000 \\
\bottomrule
\end{tabular}
\end{table*}

\subsubsection{Baselines.}
We compare RecSoup with baselines from three paradigms. \textbf{Retrain} independently optimizes each preference and serves as an oracle.

\textbf{Post-Processing.} \textbf{MMR}~\cite{carbonell1998use} re-ranks candidates to balance accuracy and diversity.

\textbf{Parameter Generation.} \textbf{CMR}~\cite{chen2023controllable} uses a hyper-network to generate preference-conditioned parameters, while \textbf{Paragon}~\cite{shen2025paragon} generates task-specific adapters with a diffusion model.

\textbf{Model Merging.} \textbf{Soup}~\cite{rameRewardedSoups2023} linearly interpolates objective-specific models, while \textbf{TATR}~\cite{sun2025tatr} mitigates inter-objective interference through orthogonal trust regions. We further include \textbf{TIES}~\cite{yadav2023ties}, a conflict-aware pruning baseline that prunes parameters based on objective-specific importance and resolves sign conflicts. In contrast, RecSoup additionally considers cross-objective relevance during pruning; detailed comparisons are presented in Sec.~\ref{rq2}. We also evaluate recent pruning-based methods, including \textbf{DARE}~\cite{yu2024dare}, \textbf{CABS}~\cite{yang2025cabs}, and \textbf{LwPTV}~\cite{zhang2025lwptv}, with additional results reported in Appendix~\ref{app:other_exp}.

\subsubsection{Implemetation Details.}
Following \text{Paragon}~\cite{shen2025paragon}, the Accuracy specialist uses the BPR loss~\cite{rendle2009bpr}, while the Diversity specialist uses a differentiable $\alpha$-NDCG-based surrogate promoting category novelty and reducing redundancy~\cite{shen2025paragon}. They form the \textbf{Retrain} Pareto-front endpoints at $\lambda=0$ and $\lambda=1$. We adopt \textbf{SASRec}~\cite{kang2018sasrec}, \textbf{GRU4Rec}~\cite{hidasi2015gru4rec}, and \textbf{TiSASRec}~\cite{li2020tisasrec}, trained with Adam, learning rate $10^{-3}$, and 64-dimensional embeddings and hidden states. We use 9 negatives for training and 99 for evaluation, set $\beta=0.5$ in Eq.~\eqref{eq:final_score}, and set $\rho\in[0.9,0.95]$. Since $\rho$ controls the pruning criterion rather than the exact retention fraction, the resulting keeping ratio is typically 0.4--0.7. Further MMR settings and coefficient-learner architectures are given in Appendix~\ref{app:implement}.
\subsubsection{Metrics.} 
We evaluate controllable multi-objective recommendation using Hypervolume, Pareto front length, and controllability. \textbf{Hypervolume (HV)}~\cite{guerreiro2021hypervolume} measures the volume of the objective space dominated by the Pareto-optimal solutions with respect to a reference point. A larger HV indicates better overall trade-offs across objectives. 
\textbf{Length of Pareto Front (LF)}~\cite{Riquelme2015LF} denotes the number of distinct Pareto-optimal solutions, reflecting the range of achievable trade-offs. 
\textbf{Controllability}~\cite{xie2025bonesoup} measures how well the model aligns with user preferences. The controllability score $C$ is defined as:  
\begin{equation*}
C = \frac{1}{M(M-1)} \sum_{i \neq j} 
\mathbb{I} \Bigg[
\bigwedge_{k=1}^{N}
\operatorname{sgn}(\mu_{i,k} - \mu_{j,k})
=
\operatorname{sgn}\big(m_k(\mathcal{S}_i) - m_k(\mathcal{S}_j)\big)
\Bigg], 
\end{equation*}
where $M$ is the number of preference vectors, $N$ the number of objectives, 
$\boldsymbol{\mu}_i$ the $i$-th preference vector, $\mathcal{S}_i$ the corresponding model output, 
$m_k(\mathcal{S}_i)$ the $k$-th metric value, and $\mathbf{1}[\cdot]$ the indicator function; 
higher $C$ indicates stronger alignment with preferences.

% For an overall evaluation of the Pareto front, we calculate the \textbf{Hypervolume (HV)}~\cite{guerreiro2021hypervolume}, which measures the area covered by the Pareto optimal points relative to a reference point. A larger HV indicates a superior overall trade-off. Additionally, we report the \textbf{Length of Pareto Front}, defined as number of distinct Pareto-optimal solutions identified along the front, to quantify the range of achievable trade-offs.
% Following~\cite{shen2024survey}, we evaluate \textbf{Controllability} to measure how well the model aligns with user-specified preferences.
\begin{figure*}[t]
    \centering
    \subfigure[Parameter decoupling.]{
      \includegraphics[width=0.24\linewidth]{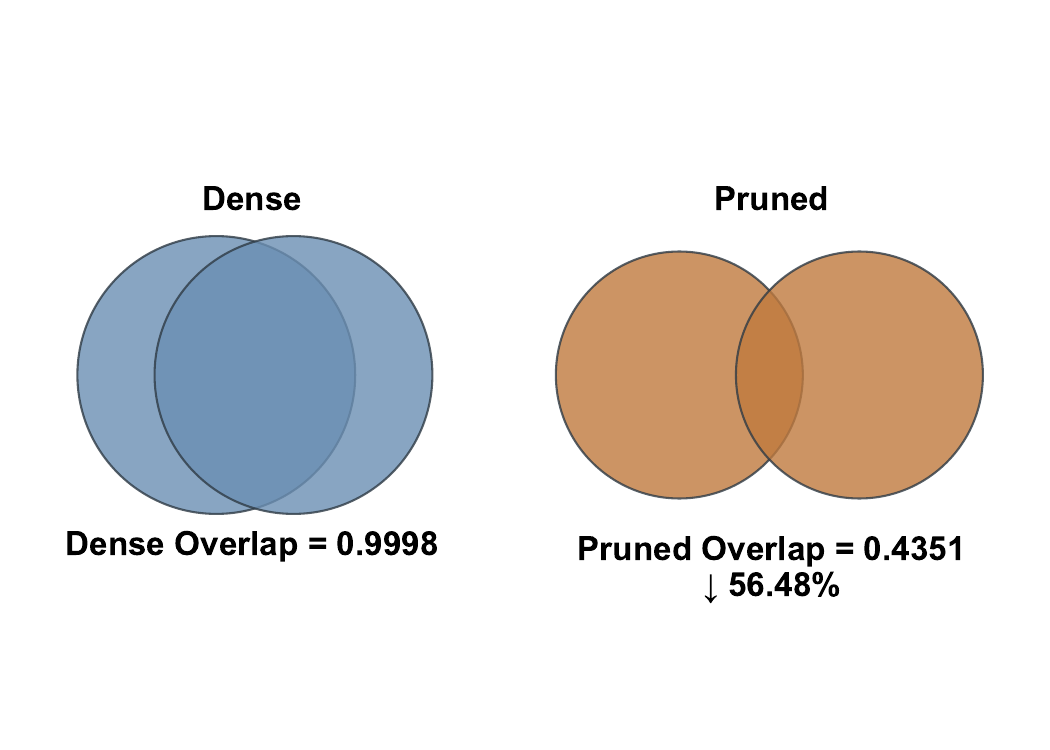}
    }\hfill
    \subfigure[Geometry preservation.]{
        \includegraphics[width=0.24\linewidth]{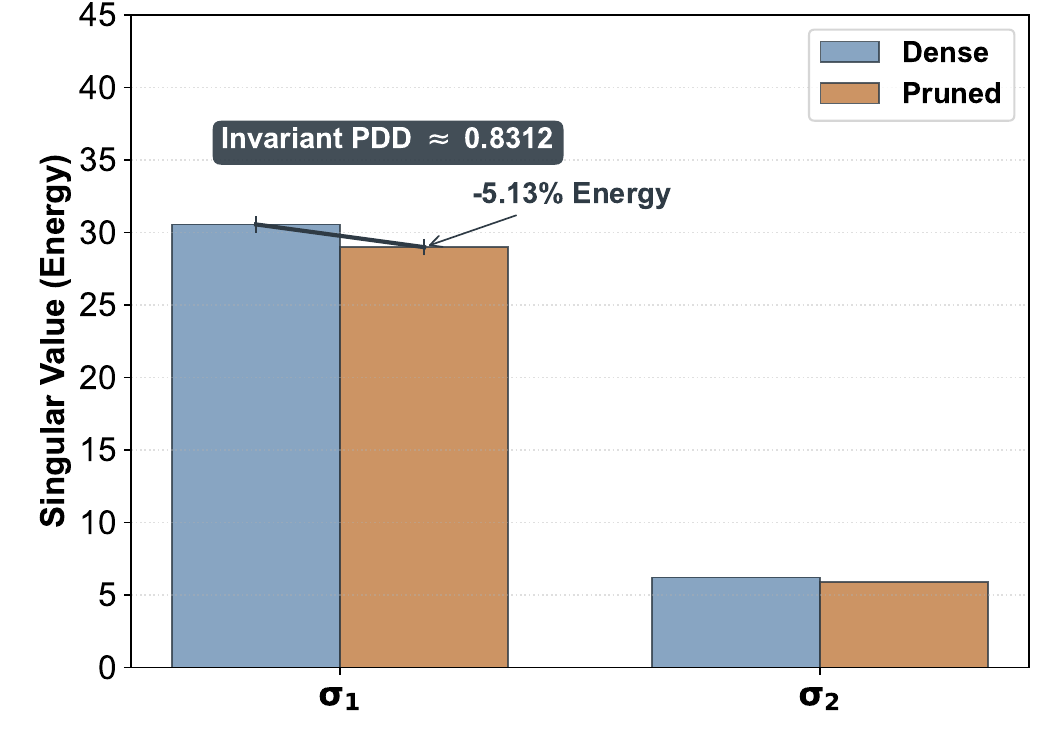}
    }\hfill
    \subfigure[Feature sharpening effect.]{
        \includegraphics[width=0.24\linewidth]{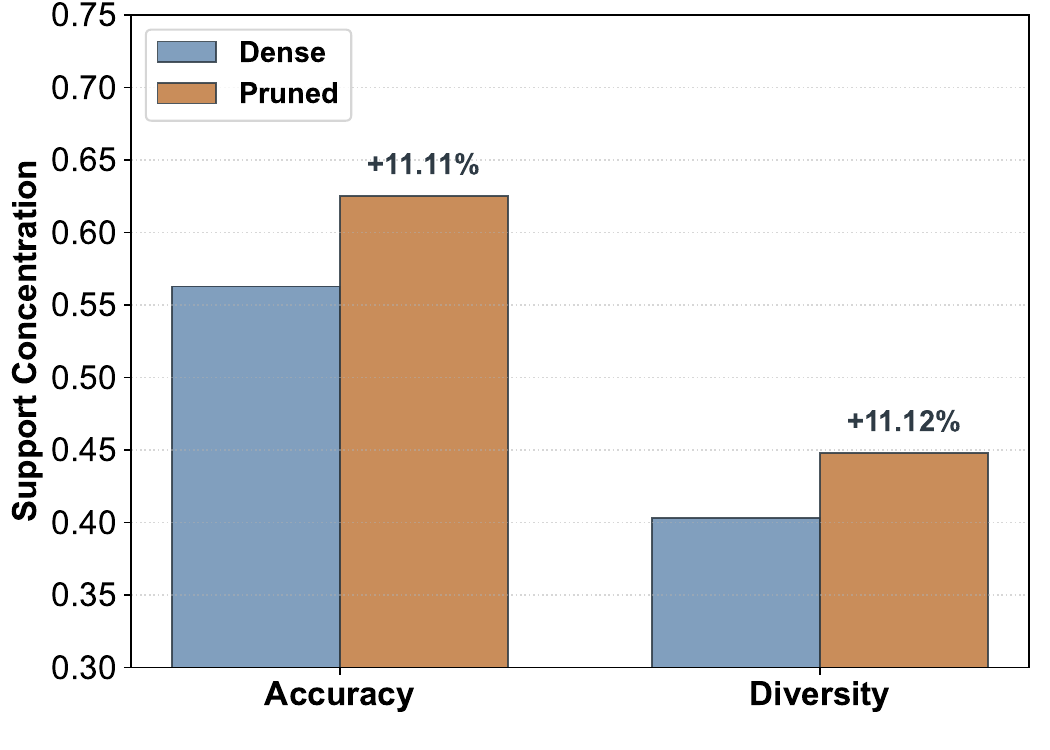}
    }\hfill
    \subfigure[Pareto front.]{
        \includegraphics[width=0.24\linewidth]{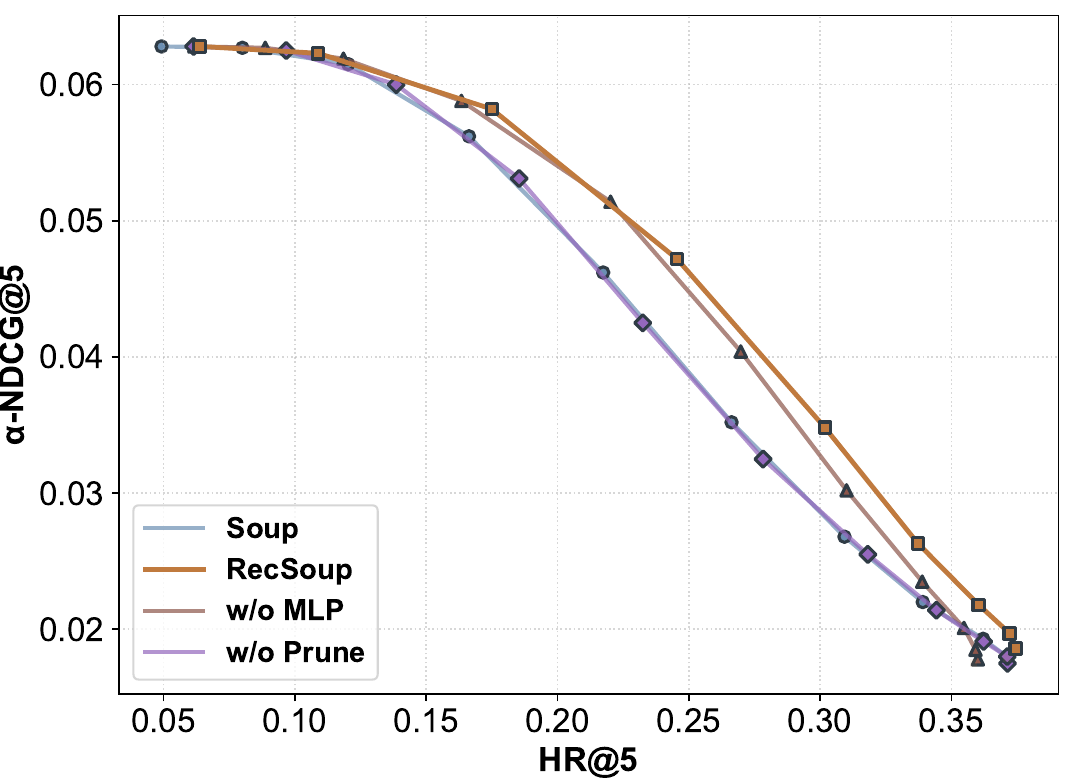}
    }

     \caption{Structural analysis of task vectors under RecSoup pruning.
    (a) Pruning reduces the overlap between objectives, indicating effective parameter decoupling.
    (b) The singular spectrum and main directions remain nearly unchanged, suggesting that the overall structure is preserved.
    (c) Support concentration increases after pruning, showing that noisy updates are removed and task-specific signals become clearer.
    (d) The resulting Pareto front is expanded.}
    \label{fig:structure}
\end{figure*}

\subsection{Experiment Results}
\label{sec:experiment-result}
We conduct experiments to evaluate RecSoup from two aspects: 
(i) its overall performance and generalization compared with baselines across different backbones and datasets;
(ii) the empirical contribution of each individual component in RecSoup.
The results are presented in Table~\ref{tab:overall_performance}.

Overall, across three backbones and datasets, RecSoup achieves a favorable balance between Pareto efficiency and controllability. Compared with strong model-merging baselines such as Soup and TATR, RecSoup consistently attains higher Hypervolume. On MovieLens-1M, where Soup and TATR already outperform earlier controllable approaches, RecSoup further improves Hypervolume by up to 8.5\% on SASRec and 7.3\% on TiSASRec, while maintaining controllability close to the upper bound of 1.0. These results suggest that linear task-vector merging alone is insufficient to fully resolve complex interactions among conflicting objectives, whereas RecSoup more effectively explores the Pareto frontier without sacrificing controllability. 
MMR exhibits a different trade-off behavior: although it achieves competitive Hypervolume under certain settings, this often comes with much shorter Pareto fronts, indicating that its gains are concentrated on a limited number of operating points. In contrast, RecSoup consistently produces longer Pareto fronts, offering a richer set of controllable solutions, enabling flexible selection for practical use.
Parameter generation approaches exhibit unstable behavior: despite high controllability, their Hypervolume degrades sharply under sparse or noisy feedback, particularly on Amazon. In these settings, RecSoup outperforms them with around 20--50\% higher Hypervolume in most cases while preserving controllability above 0.89. Notably, RecSoup occasionally surpasses the oracle, likely due to its ability to explore feasible Pareto solutions beyond the optimization paths induced by retraining. %Moreover, RecSoup attains the highest or second-highest Hypervolume even on lower-capacity backbones such as GRU4Rec, indicating strong generalization.

The ablation results show that adaptive merging is necessary for stable controllable recommendation. Removing coefficient learner consistently reduces Hypervolume, which indicates that directly using the preference vector as merging coefficients cannot sufficiently capture their relationship. The role of pruning is more conditional. In several settings, w/o Prune performs comparably to or even better than RecSoup, suggesting that pruning may be unnecessary when task vectors are already compatible or when sparsification removes useful objective-specific updates. However, pruning becomes important when task vectors contain stronger interference or noisy updates, especially on more challenging Amazon settings. By selecting parameters according to both objective-specific importance and cross-objective relevance, pruning provides a cleaner basis for adaptive merging. Thus, pruning mainly improves robustness under conflicting objectives. 
%Overall, these results indicate that pruning and adaptive merging are both individually necessary for achieving stable and high-quality controllable trade-offs, and removing either component leads to noticeable performance degradation under certain settings.
\begin{figure}[ht]
    \centering
    \subfigure[Accuracy Preference Alignment.]{
        \includegraphics[width=0.45\linewidth]{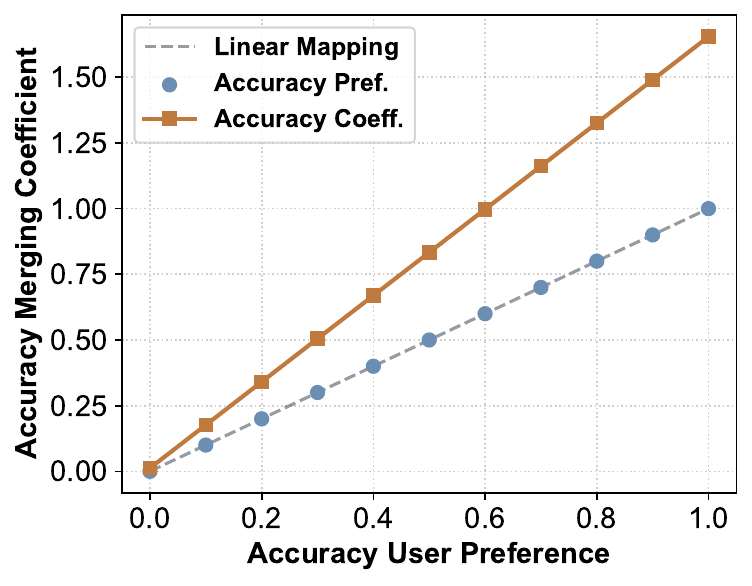}}
    \hfill
    \subfigure[Diversity Preference Alignment.]{
        \includegraphics[width=0.45\linewidth]{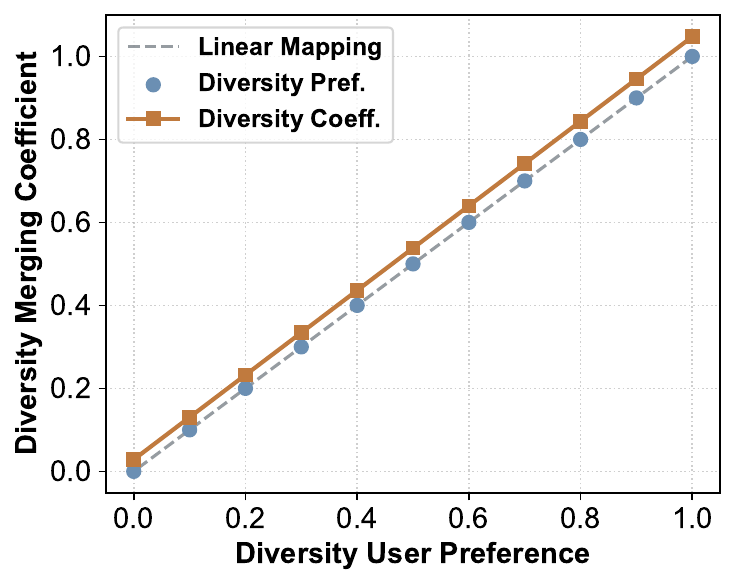}}
    \caption{Alignment between user preferences and learned merging coefficients.}
    \label{fig:preference_alignment}
\end{figure}

\subsection{Analyses}
\subsubsection{\textbf{RQ1: How do Two-stage Design Facilitate Controllable Multi-Objective Recommendation?}}
\label{sec:synergy}
Figure~\ref{fig:structure} explains why pruning synergizes with adaptive merging.
Figure~\ref{fig:structure}(a) reports the support overlap ratio between pairs of task vectors, computed as the intersection-over-union of parameter supports whose absolute update magnitude exceeds a small threshold.
Pruning clearly reduces this overlap, indicating effective parameter decoupling across objectives. By reducing shared and potentially conflicting parameters, pruning lowers cross-objective interference and provides a cleaner basis for multi-objective merging. This decoupling does not largely change the overall structure of task vectors.
As shown in Figure~\ref{fig:structure}(b), the singular spectrum and main directions remain nearly unchanged after pruning, suggesting that pruning preserves the main structure of parameter updates.
Figure~\ref{fig:structure}(c) further shows that support concentration increases after pruning, meaning that objective-specific updates become more focused on a smaller set of important parameters.
Together, these changes make the task-vector space cleaner and more controllable, which leads to an expanded Pareto front in Figure~\ref{fig:structure}(d). Based on this pruned parameter space, Figure~\ref{fig:preference_alignment} further analyzes the learned merging coefficients.
A simple linear rule assumes that the input preference vector can be directly used as the merging coefficients. This assumption can be violated when objectives have different parameter scales, non-linear interactions, or residual conflicts. RecSoup addresses this linear regime violation by decoupling parameter selection from coefficient learning: pruning first builds a cleaner task-vector basis, and adaptive merging then learns non-linear coefficients conditioned on preferences. Thus, RecSoup can amplify or suppress each objective beyond the input preference value, enabling more flexible navigation along the Pareto front.

\begin{figure}[t]
  \centering
 \includegraphics[width=\linewidth]{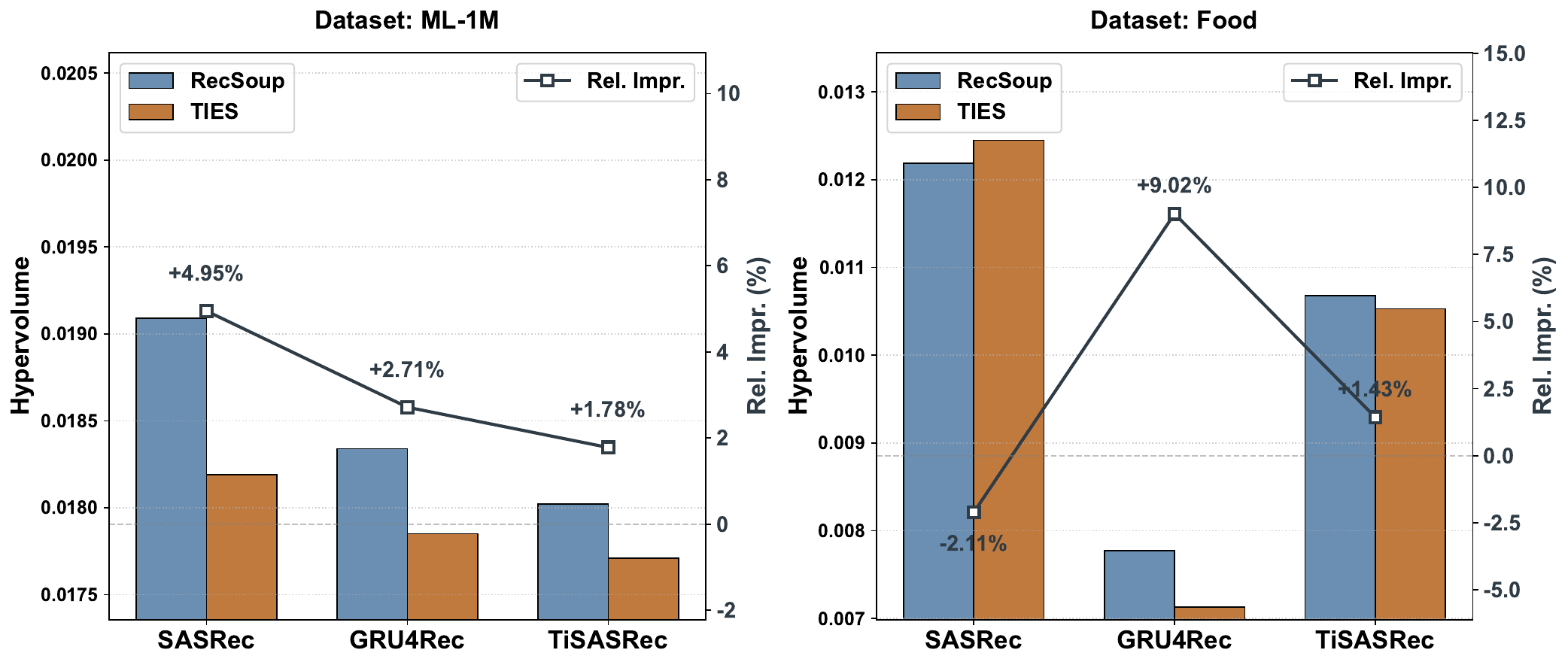}
  \caption{Hypervolume (HV) comparison between RecSoup and TIES. Bars show absolute HV, and the line indicates the relative improvement of RecSoup over TIES.}
  \Description{A dual-axis chart showing that RecSoup consistently achieves higher Hypervolume scores compared to TIES across SASRec, GRU4Rec, and TiSASRec. The relative improvement reaches up to 9.02\% on the Food dataset, demonstrating the effectiveness of the proposed weight soup strategy.}
  \label{fig:ties-vs-sim}
\end{figure}
\subsubsection{\textbf{RQ2: Impact of Cross-Objective Relevance Score}}
\label{rq2}
We investigate the contribution of the Cross-Objective Relevance Score $\mathbf{T}_{\mathrm{align}}$ defined in Eq.~\eqref{eq:align_score} to the pruning effectiveness of RecSoup. The pruning strategy in RecSoup is guided by two complementary signals: $\mathbf{T}_{\mathrm{self}}$  
%which measures the importance of individual objectives, 
and $\mathbf{T}_{\mathrm{align}}$. 
%which captures the cross-objective relevance. 
In contrast, TIES performs pruning solely based on $\mathbf{T}_{\mathrm{self}}$. To isolate the impact of $\mathbf{T}_{\mathrm{align}}$, we ensure that both RecSoup and TIES share identical merging coefficients and merging procedures. Figure~\ref{fig:ties-vs-sim} presents the Hypervolume (HV) results and the relative improvement of RecSoup over the TIES baseline across three backbones and multiple datasets. On MovieLens-1M, RecSoup consistently outperforms TIES, achieving relative HV gains of +4.95\% (SASRec), +2.71\% (GRU4Rec), and +1.78\% (TiSASRec). These results confirm that $\mathbf{T}_{\mathrm{align}}$ effectively guides the pruning process: by considering cross-objective relevance in addition to objective-specific importance, RecSoup selectively retains parameters critical to multiple objectives. This prevents the over-pruning of parameters that are essential for one objective but not heavily weighted by $\mathbf{T}_{\mathrm{self}}$, which TIES may discard. In summary, this ablation demonstrates that the similarity score is a key component in RecSoup's pruning stage, directly contributing to improved multi-objective performance by preserving cross-task compatibility while maintaining a controlled merging process. 
%The observed improvements across multiple datasets and backbone architectures highlight its robustness.

\subsubsection{\textbf{RQ3: Inference-Time Efficiency of Parameter Generation and RecSoup}}We compare hypernetwork-based generation approach~\cite{chen2023controllable} with RecSoup in terms of their wall-clock latency following the time complexity analysis in Section~\ref{sec:complexity analyses}.
Table~\ref{tab:timing} reports the average parameter construction time over 11 user preferences across different datasets using the same backbone. The results show that RecSoup consistently incurs lower latency than hypernetwork-based generation. This efficiency advantage is particularly pronounced when multiple preference configurations are evaluated, as merging reuses a fixed set of task vectors and avoids repeated parameter synthesis. In contrast, diffusion-based parameter generation is substantially slower due to its iterative nature. In our experiments, the diffusion process was configured with 10,000 denoising steps; and under this setting, the average parameter construction time per preference was approximately 87.79 seconds, orders of magnitude larger than both hypernetwork generation and RecSoup. These observations validate our theoretical analysis and demonstrate that RecSoup provides a practical and efficient solution for inference-time controllable recommendation.

\begin{table}[t]
\centering
\caption{Inference-time parameter construction latency (milliseconds) for different controllable approaches.
All results are averaged over 10 runs.}
\label{tab:timing}
\begin{tabular}{lccc}
\toprule
Dataset & Hypernetwork & RecSoup & Speedup \\
\midrule
ML-1M & 7.5281 $\pm$ 0.5990 & 0.0454 $\pm$ 0.0100 & 165.9× \\
Amazon(Food) & 9.3799 $\pm$ 0.1756 & 0.0410 $\pm$ 0.0011& 228.8× \\
Amazon(Games) & 9.2741 $\pm$ 0.2746 & 0.0436 $\pm$ 0.0055 & 212.7× \\
\bottomrule
\end{tabular}
\end{table}

\begin{table}[t]
\centering
\caption{Accuracy-Fairness controllable evaluation results on ML-1M with SASRec.}
\label{tab:rq4_acc_fair}
\renewcommand{\arraystretch}{0.92}
\setlength{\tabcolsep}{2.5pt}
\begin{tabular}{lccc|lccc}
\toprule
\multicolumn{4}{c|}{Baselines} & \multicolumn{4}{c}{Our Method} \\
\cmidrule(lr){1-4} \cmidrule(lr){5-8}
Method & HV$\uparrow$ & Contro.$\uparrow$ & LF$\uparrow$
& Method & HV$\uparrow$ & Contro.$\uparrow$ & LF$\uparrow$ \\
\midrule
Soup & 0.3968 & 0.6909 & 6
& \cellcolor{gray!10}\textbf{RecSoup} & \cellcolor{gray!10}\textbf{0.4270} & \cellcolor{gray!10}0.6727 & \cellcolor{gray!10}5 \\

TATR & 0.2064 & 0.7818 & 6
& w/o MLP & \underline{0.4176} & 0.7091 & 5 \\

TIES & 0.3858 & 0.7818 & 6
& w/o Prune & 0.3872 & 0.7818 & 6 \\

CMR & 0.2643 & \textbf{0.8364} & 4
&  &  &  &  \\
\bottomrule
\end{tabular}
\end{table}

\subsubsection{\textbf{RQ4: How does RecSoup scale to additional objectives such as fairness?}}
We first extend controllable recommendation to the Accuracy-Fairness setting by constructing Accuracy-oriented and Fairness-oriented specialist models. Fairness is measured using the absolute NDCG gap between male and female user groups.

Table~\ref{tab:rq4_acc_fair} reports the controllable results on MovieLens-1M with SASRec. RecSoup achieves the highest HV among all compared controllable baselines, showing stronger Pareto efficiency in the Accuracy-Fairness setting. The performance drop of ``w/o MLP'' and ``w/o Prune'' further verifies the contribution of the coefficient learner and the pruning strategy. We exclude Paragon from the main comparison because it collapses to a single Pareto solution in this setting, which cannot provide meaningful controllable transitions across different preferences. Overall, these results show that RecSoup can be effectively generalized to fairness-aware controllable recommendation.

We further extend controllable recommendation to jointly optimize Accuracy, Diversity, and Fairness. As shown in Table~\ref{tab:rq4_three_objective_main}, RecSoup achieves the best Hypervolume and overall controllability among all baselines, while also obtaining the strongest controllability on the Accuracy--Diversity objective pair. In contrast, CABS exhibits a substantial degradation in Hypervolume despite relatively strong controllability on the Accuracy--Fairness pair, suggesting that strong local controllability does not necessarily translate into globally balanced Pareto frontiers under multi-objective settings. Although Soup produces slightly more Pareto solutions, its lower HV indicates inferior frontier quality compared with RecSoup. Overall, these results demonstrate the scalability of RecSoup in multi-objective recommendation settings.
\begin{table}[t]
\centering
\caption{Three-objective controllable recommendation results on ML-1M with SASRec. Control. represents the averaged controllability score across objective pairs.}
\label{tab:rq4_three_objective_main}
\begin{tabular*}{\columnwidth}{@{\extracolsep{\fill}}lccccc}
\toprule
Method 
& HV$\uparrow$ & Control.$\uparrow$ & Acc--Div$\uparrow$ & Acc--Fair$\uparrow$ & LF$\uparrow$ \\
\midrule

SOUP    & \underline{0.6015} & 0.6479 & \underline{0.9152} & \underline{0.8816} & \textbf{49} \\

TATR    & 0.5010 & 0.6211 & 0.8732 & 0.8135 & 31 \\

DARE    & 0.5881 & \underline{0.6482} & 0.9147 & 0.8639 & 45 \\

CABS    & 0.2594 & 0.6395 & 0.4172 & \textbf{0.9207} & 42 \\

LwPTV   & 0.5991 & 0.6452 & 0.9119 & \underline{0.8816} & 42 \\

\midrule

RecSoup & \textbf{0.6148} & \textbf{0.6513} & \textbf{0.9245} & 0.8643 & \underline{46} \\

\bottomrule
\end{tabular*}
\end{table}

\section{Conclusions}
This work shows that task-vector merging, despite its success in large language models, cannot be directly applied to multi-objective sequential recommendation due to the absence of a shared linear regime. Motivated by this observation, we propose \textbf{RecSoup}, which mitigates parameter interference through pruning and enables controllable trade-offs via a lightweight gating network. Empirical results further indicate that model-merging approaches offer a more efficient and effective paradigm than parameter-generation approaches for controllable multi-objective recommendation.

\begin{acks}
This work was partially supported by the National Natural Science Foundation of China (No. 62376275, 62472426). Work partially done at Beijing Key Laboratory of Research on Large Models and Intelligent Governance, and Engineering Research Center of Next-Generation Intelligent Search and Recommendation, MOE. Supported by fund for building world-class universities (disciplines) of Renmin University of China.
\end{acks}

\appendix
\section{Proof of Theorem~\ref{thm:recmerge:solutionspace}}
\label{app:proof}
\begin{proof}[Proof of Theorem~\ref{thm:recmerge:solutionspace}]
We prove the conclusion by showing that there exists at least one element in \( P \) that is not in the image of \( g \).

First, consider the midpoint of \( P \):
$
\theta^* = \frac{1}{2}\theta_1 + \frac{1}{2}\theta_2 \in P. 
$
We aim to show \( \theta^* \notin g([0,1]) \) via contradiction. Assume for contradiction that \( \theta^* \in g([0,1]) \), so there exists \( \lambda^* \in [0,1] \) such that \( g(\lambda^*) = \theta^* \).

By the Lipschitz continuity of \( g \), we have:
\[
\|\theta^* - g(1)\|_2 \leq L \cdot |\lambda^* - 1| \quad \text{and} \quad \|\theta^* - g(0)\|_2 \leq L \cdot |\lambda^*|.
\]
Summing these two inequalities and using \( |\lambda^* - 1| + |\lambda^*| = 1 \) (since \( \lambda^* \in [0,1] \)) gives
$
\|\theta^* - g(1)\|_2 + \|\theta^* - g(0)\|_2 \leq L.
$

Next, analyze the practical utility of \( g(1) \) and \( g(0) \). For \( \lambda = 1 \), \( g(1) \) must be optimized for the same objective as \( \theta_1 \) (to be practically useful for multi-objective trade-off), so \( g(1) = \theta_1 \) (otherwise, it cannot serve as the extreme point of the objective trade-off). Similarly, \( g(0) = \theta_2 \) for \( \lambda = 0 \).

By the triangle inequality, 
$$
L_0 = \|\theta_1 - \theta_2\|_2 \leq \|\theta_1 - \theta^*\|_2 + \|\theta^* - \theta_2\|_2.
$$ 
Since \( g(1) = \theta_1 \) and \( g(0) = \theta_2 \), this simplifies to:
\[
\|\theta^* - g(1)\|_2 + \|\theta^* - g(0)\|_2 = \|\theta_1 - \theta^*\|_2 + \|\theta^* - \theta_2\|_2 = L_0.
\] 
Combining the two results above yields \( L_0 \leq L \), which contradicts the given condition \( L < L_0 \). Thus, our initial assumption is invalid, meaning \( \theta^* \notin g([0,1]) \). 
Since \( \theta^* \in P \) but \( \theta^* \notin g([0,1]) \), the parameter generation mapping \( g \) cannot generate all elements of the Pareto-optimal set \( P \).
\end{proof}

\section{Additional Implementation Details}
\label{app:implement}
\paragraph{MMR Evaluation Settings.}
For MMR evaluation, the recommendation model first generates a candidate pool containing the top-100 items for each user, and MMR reranking is applied to produce the final top-100 recommendation list. The relevance term is aligned with the accuracy objective, while the diversity term is computed using Jaccard dissimilarity over category-level one-hot item encodings. We evaluate 11 controllable preference vectors satisfying $\mu_{acc}+\mu_{div}=1$, where $\mu_{acc}\in\{0.0,0.1,\dots,1.0\}$.

\paragraph{Coefficient Learner Architecture.}
The coefficient learner is implemented as a lightweight routing module conditioned on the preference vector. In all experiments, we use a single-layer linear router. The routing weights are initialized using a Gaussian distribution with standard deviation 0.01. During training, the backbone model and expert task vectors are frozen, and only the routing parameters are updated. The coefficient learner is optimized using Adam with learning rate 0.001, batch size 512.
\begin{table}[t]
\centering
\caption{Overall comparison with additional pruning-based merging baselines. All metrics are averaged over Accuracy-Diversity configurations.}
\label{tab:additional_baselines_avg}
\begin{tabular}{lcccc}
\toprule
Method & Avg.HV$\uparrow$ & Avg.Contro.$\uparrow$ & Avg.LF$\uparrow$
& Avg.Rank$\downarrow$ \\
\midrule
DARE  & 0.4977 & 0.9414 & 9.56 & 3.67 \\
CABS  & 0.5044 & 0.8951 & 9.44 & 2.56 \\
LwPTV & 0.5200 & 0.9111 & 9.56 & 2.22 \\
\rowcolor{gray!10}RecSoup & \textbf{0.5662} & \textbf{0.9657} & \textbf{9.89} & \textbf{1.56} \\
\bottomrule
\end{tabular}
\end{table}

\section{Additional Experimental Results}
\label{app:other_exp}
Table~\ref{tab:additional_baselines_avg} summarizes the overall comparison with additional pruning-based merging baselines across all backbone-dataset configurations under the Accuracy-Diversity setting. \textbf{DARE}~\cite{yu2024dare} randomly drops task-vector parameters followed by rescaling. \textbf{CABS}~\cite{yang2025cabs} performs sequential sparse pruning to reduce overlap among task vectors, and \textbf{LwPTV}~\cite{zhang2025lwptv} prunes layer-wise salient parameters with low inter-task deviation. RecSoup achieves the highest average hypervolume, controllability, and Pareto-front coverage, while also obtaining the best average ranking among all methods. In particular, RecSoup consistently outperforms DARE in all settings and demonstrates stronger overall Pareto efficiency than CABS and LwPTV, indicating that incorporating cross-objective relevance into the pruning process leads to more robust controllable model merging.

\section*{GenAI Usage Disclosure}
The authors used generative AI tools only for language polishing and improving clarity of expression. All AI-assisted edits were reviewed and verified by the authors, who take full responsibility for the final content.
%%
%% The next two lines define the bibliography style to be used, and
%% the bibliography file.
\bibliographystyle{ACM-Reference-Format}
\bibliography{reference}

\end{document}